\documentclass[12pt]{article}
\usepackage{amsmath}
\usepackage{graphicx}
\usepackage{amsfonts}                       % blackboard math symbols
\usepackage{mathtools}
\usepackage{bm}
\usepackage{placeins}
\usepackage{enumerate}
\usepackage{natbib}
\usepackage{xcolor}
\usepackage{filecontents}
\usepackage{url} % not crucial - just used below for the URL
\usepackage[margin=1in]{geometry}
\usepackage[font=small]{caption}

\usepackage{xr} % nameref before zref-xr
\DeclareMathOperator*{\argmin}{\arg\!\min}
\newtheorem{assumption}{Assumption}
\newtheorem{theorem}{Theorem}
\newtheorem{corollary}{Corollary}
\newtheorem{proposition}{Proposition}
\newcommand{\rone}[1]{#1}
\newcommand{\rtwo}[1]{#1}

\newcommand{\blind}{0}

\begin{document}

\def\spacingset#1{\renewcommand{\baselinestretch}%
{#1}\small\normalsize} \spacingset{1}

%%%%%%%%%%%%%%%%%%%%%%%%%%%%%%%%%%%%%%%%%%%%%%%%%%%%%%%%%%%%%%%%%%%%%%%%%%%%%%

\if0\blind
{
  \title{\bf Cross-Validated Loss-Based Covariance Matrix Estimator Selection
  in High Dimensions}
\author{
  {\normalsize Philippe Boileau}\\
    {\normalsize Graduate Group in Biostatistics and Center for Computational
     Biology, UC Berkeley} \\
  {\normalsize Nima S.\ Hejazi}\\
    {\normalsize Division of Biostatistics, Department of Population Health
        Sciences, Weill Cornell Medicine}\\
  {\normalsize Mark J.~van der Laan}\\
    {\normalsize Division of Biostatistics, Department of Statistics,}\\
    {\normalsize and Center for Computational Biology, UC Berkeley}\\
  {\normalsize Sandrine Dudoit}\thanks{To whom correspondence should be
    addressed: sandrine@stat.berkeley.edu}\hspace{0.2cm}\\
    {\normalsize Department of Statistics, Division of Biostatistics,}\\
    {\normalsize and Center for Computational Biology, UC Berkeley}
}
  \maketitle
} \fi

\if1\blind
{
  \bigskip
  \bigskip
  \bigskip
  \begin{center}
    {\LARGE\bf Cross-Validated Loss-Based Covariance Matrix Estimator Selection
      in High Dimensions}
\end{center}
  \medskip
} \fi
\vspace{-1cm}
\begin{abstract}
  The covariance matrix plays a fundamental role in many modern exploratory and
  inferential statistical procedures, including dimensionality reduction,
  hypothesis testing, and regression. In low-dimensional regimes, where the
  number of observations far exceeds the number of variables, the optimality of
  the sample covariance matrix as an estimator of this parameter is
  well-established. High-dimensional regimes do not admit such a convenience.
  Thus, a variety of estimators have been derived to overcome the shortcomings
  of the canonical estimator in such settings. Yet, selecting an optimal
  estimator from among the plethora available remains an open challenge. Using
  the framework of cross-validated loss-based estimation, we develop the
  theoretical underpinnings of just such an estimator selection procedure. We
  propose a general class of loss functions for covariance matrix estimation
  and establish accompanying finite-sample risk bounds and conditions for the
  asymptotic optimality of the cross-validation selector. In numerical
  experiments, we demonstrate the optimality of our proposed selector in
  moderate sample sizes and across diverse data-generating processes. The
  practical benefits of our procedure are highlighted in a dimension reduction
  application to single-cell transcriptome sequencing data.
\end{abstract}

\noindent%
{\it Keywords:} covariance matrix estimation; cross-validation; dimension
reduction; high-dimensional statistics; loss-based estimation 

% \noindent%
% {\footnotesize\textit{Word count,} \texttt{main.text}\textit{:} 5508}
% \vfill

\newpage
\spacingset{1.5} % DON'T change the spacing!

\makeatletter\@input{get-os-aux.tex}\makeatother

\section{Introduction}
\label{intro}

The covariance matrix underlies numerous exploratory and inferential statistical
procedures central to analyses regularly performed in diverse fields. For
instance, in computational biology, this statistical parameter serves as a key
ingredient in many popular dimensionality reduction, clustering, and
classification methods which are regularly relied upon in quality control
assessments, exploratory data analysis, and, recently, the discovery and
characterization of novel types of cells. Other important areas in which the
covariance matrix is critical include financial economics, climate modeling and
weather forecasting, imaging data processing and analysis, probabilistic
graphical modeling, and text corpora compression and retrieval. Even more
fundamentally, the covariance matrix plays a key role in assessing the strengths
of linear relationships within multivariate data, in generating simultaneous
confidence bands and regions, and in the construction and evaluation of
hypothesis tests. Accurate estimation of this parameter is therefore essential.

When the number of observations in a data set far exceeds the number of
features, the estimator of choice for the covariance matrix is the sample
covariance matrix: it is an efficient estimator under minimal regularity
assumptions on the data-generating distribution
\rone{\citep{anderson2003,smith2005}}. In high-dimensional regimes, however,
this simple estimator has undesirable properties. When the number of features
outstrips the number of observations, the sample covariance matrix is singular.
Even when the number of observations slightly exceeds the number of features,
the sample covariance matrix is numerically unstable on account of an overly
large condition number~\citep{golub1996matrix}. Its eigenvalues are also
generally over-dispersed when compared to those of the population covariance
matrix~\citep{johnstone2001distribution,ledoit:2004}: the leading eigenvalues
are positively biased, while the trailing eigenvalues are negatively biased
\rone{\citep{marchenko1967}}.

High-dimensional data have become increasingly widespread in myriad scientific
domains, with many examples arising from challenges posed by cutting-edge
biological sequencing technologies. Accordingly, researchers have turned to
developing novel covariance matrix estimators to remediate the deficiencies of
the sample covariance matrix. Under certain sparsity assumptions,
\citet{bickela:2008,bickelb:2008,rothman:2009,lam:2009,cai:2010}, and
\citet{cai:2011}, among others, demonstrated that the true covariance matrix can
be estimated consistently under specific losses by applying element-wise
thresholding or tapering functions to the sample covariance matrix. Another
thread of the literature, which includes notable contributions by
\citet{stock:2002, bai:2003, fan:2008, fan:2013, fan:2016b, fan:2019}, and
\citet{onatski:2012}, has championed methods employing factor models in
covariance matrix estimation. Other popular proposals include the families of
estimators inspired by the empirical Bayes
framework~\citep{robbins1964empirical, efron2012large}, formulated by
\citet{shafer:2005} and \citet{ledoit:2004, ledoit:2012, ledoit:2015,
ledoit:2020b}. We briefly review several of these estimator families in
Section~\ref{supp:lit} of the Online Supplement.

\rtwo{Despite the flexibility afforded by the apparent wealth of candidate
estimators, this variety poses many practical issues. Namely, identifying the
most appropriate estimator from among a collection of candidates is itself a
significant challenge. A partial answer to this problem has come in the form
of data-adaptive approaches designed to select the optimal estimator within a
particular class~\rone{\citep[for
example,][]{bickela:2008,bickelb:2008,cai:2011,fan:2013,fang:2015,bartz2016}}.
Such approaches, however, are inherently limited by their focus on relatively
narrow families of covariance matrix estimators. The successful application
of such estimator selection frameworks requires, as a preliminary step, that
the practitioner make a successful choice among estimator families, injecting
a degree of subjectivity in their deployment. The broader question of
selecting an optimal estimator from among a diverse library of candidates has
remained unaddressed. We offer a general loss-based framework building upon
the seminal work of \citet{laan-dudoit:2003b, dudoit2005asymptotics,
vaart:2006} for asymptotically optimal covariance matrix estimator selection
based upon cross-validation.}

In the cross-validated loss-based estimation framework, the parameter of
interest is defined as the risk minimizer, with respect to the data-generating
distribution, based on a loss function chosen to reflect the problem at hand.
Candidate estimators may be generated in a variety of manners, including as
empirical risk minimizers with respect to an empirical distribution over
parameter subspaces corresponding to models for the data-generating
distribution. One would ideally select as optimal estimator that which
minimizes the ``true'' risk with respect to the data-generating distribution.
However, as this distribution is typically unknown, one turns to
cross-validation for risk estimation. \citet{laan-dudoit:2003b,
dudoit2005asymptotics, vaart:2006} have shown that, under general conditions on
the data-generating distribution and loss function, the cross-validated
estimator selector is asymptotically optimal in the sense that it performs
asymptotically as well in terms of risk as an optimal oracle selector based on
the true, unknown data-generating distribution. \rone{These results extend prior
work summarized by~\citet[][Ch.~7--8]{gyorfi2002}.}

\rtwo{Focusing specifically on the covariance matrix as the parameter of
interest, we address the choice of loss function and candidate estimators,
and derive new, high-dimensional asymptotic optimality results for
multivariate cross-validated estimator selection procedures. Requiring
generally nonrestrictive assumptions about the structure of the true
covariance matrix, the proposed framework accommodates arbitrary families of
covariance matrix estimators. The method therefore enables the objective
selection of an optimal estimator while fully taking advantage of the
plethora of available estimators. As such, it generalizes existing, but more
limited, data-adaptive estimator selection frameworks where the library of
candidate estimators is narrowed based on available subject matter knowledge,
or, as is more commonly the case, for convenience's sake.}

\section{Problem Formulation and Background}
\label{backg}

Consider a data set $\mathbf{X}_{n \times J} = \{X_1, \ldots, X_n: X_i \in
\mathbb{R}^{J}\}$, comprising $n$ independent and identically distributed
(i.i.d.)~random vectors, where $n \approx J$ or $n < J$. Let $X_i \sim P_0 \in
\mathcal{M}$, where $P_0$ denotes the true data-generating distribution and
$\mathcal{M}$ the statistical model, that is, a collection of possible
data-generating distributions $P$ for $X_i$. We assume a nonparametric
statistical model $\mathcal{M}$ for $P_0$, minimizing assumptions on
the form of $P_0$. We denote by $P_n$ the empirical distribution of the $n$
random vectors forming $\mathbf{X}_{n \times J}$. Letting $\mathbb{E}[X_i] = 0$
without loss of generality and defining $\boldsymbol{\psi}_0 \coloneqq
\text{Var}[X_i]$, we take as our goal the estimation of the covariance matrix
$\boldsymbol{\psi}_0$. This is accomplished by identifying the ``optimal''
estimator of $\boldsymbol{\psi}_0$ from among a collection of candidates,
where, as detailed below, optimality is defined in terms of risk.

For any distribution $P \in \mathcal{M}$, define its covariance matrix as
$\boldsymbol{\psi} = \Psi(P)$, where $\Psi$ is a mapping from the model
$\mathcal{M}$ to the set of $J \times J$ symmetric, positive semi-definite
matrices. \rone{Furthermore, candidate estimators of the covariance matrix are
  defined as $\hat{\boldsymbol{\psi}}_k \coloneqq \hat{\Psi}_k(P_n)$ for $k =
  1, \ldots, K$ in terms of mappings $\hat{\Psi}_k$ from the empirical
  distribution $P_n$ to $\boldsymbol{\Psi} \coloneqq \{\boldsymbol{\psi} \in
\mathbb{R}^{J\times J} | \boldsymbol{\psi} = \boldsymbol{\psi}^\top\}$}. While
this notation suggests that the number of candidate estimators $K$ is fixed,
and we treat it as such throughout, this framework may be extended such that
$K$ grows as a function of $n$ and $J$. \rone{It also follows that
  $\{\boldsymbol{\psi} = \Psi(P): P \in \mathcal{M}\} \subset
  \boldsymbol{\Psi}$; that is, the set of all covariance matrices corresponding
to the data-generating distributions $P$ belonging to the model $\mathcal{M}$
is a subset of $\boldsymbol{\Psi}$.}

In order to assess the optimality of estimators in the set $\boldsymbol{\Psi}$,
we introduce a generic loss function $L(X; \boldsymbol{\psi}, \eta)$
characterizing a cost applicable to any $\boldsymbol{\psi} \in
\boldsymbol{\Psi}$ and $X \sim P \in \mathcal{M}$, and where $\eta$ is a
(possibly empty) nuisance parameter. Specific examples of loss functions for
the covariance estimation setting are proposed in Section~\ref{loss-funs}.
Define $H$ as the mapping from the model $\mathcal{M}$ to the nuisance
parameter space $\mathbf{H} \coloneqq \{\eta = H(P): P \in \mathcal{M}\}$ and
let $\hat{\eta} \coloneqq \hat{H}(P_n)$ denote a generic nuisance parameter
estimator, where $\hat{H}$ is a mapping from $P_n$ to $\mathbf{H}$. Given any
$\eta \in \mathbf{H}$, the risk under $P \in \mathcal{M}$ for any
$\boldsymbol\psi \in \boldsymbol\Psi$ is defined as the expected value of $L(X;
\boldsymbol{\psi}, \eta)$ with respect to $P$:

\begin{equation*}
  \begin{split}
    \Theta(\boldsymbol\psi, \eta, P) & \coloneqq \int L(x; \boldsymbol{\psi},
      \eta)dP(x) \\
    & = \mathbb{E}_{P} \left[L(X; \boldsymbol{\psi}, \eta)\right].
  \end{split}
\end{equation*}

Under the additional constraint on the loss function that a risk minimizer
exists under the true data-generating distribution $P_0$, the minimizer is given
by the parameter of interest
\begin{equation}\label{risk-min-def}
  \boldsymbol{\psi}_0 \coloneqq
    \argmin_{\boldsymbol{\psi} \in \mathbf{\Psi}} \Theta(\boldsymbol\psi, \eta_0, P_0),
\end{equation}
where $\eta_0 \coloneqq H(P_0)$. The risk minimizer need not be unique. The
optimal risk under $P_0$ is
\begin{equation*} \theta_0 \coloneqq
  \min_{\boldsymbol{\psi} \in \mathbf{\Psi}} \; \Theta(\boldsymbol\psi, \eta_0,
  P_0),
\end{equation*}
which is to say that a risk minimizer $\boldsymbol\psi_0$ attains risk
$\theta_0$.

For any given estimator $\hat{\boldsymbol{\psi}}_k$ of $\boldsymbol{\psi}_0$,
its conditional risk given $P_n$ with respect to the true data-generating
distribution $P_0$ is
\begin{equation*}
  \begin{split}
    \tilde{\theta}_n(k, \eta_0) & \coloneqq \mathbb{E}_{P_0}
      [L(X; \hat{\Psi}_k(P_n), \eta_0) \mid P_n] \\
%       & = \int L(x; \hat{\Psi}_k(P_n), \eta_0)dP_0(x)\\
    & = \Theta(\hat{\boldsymbol{\psi}}_k, \eta_0, P_0).
  \end{split}
\end{equation*}
Defining the risk difference of the $k^\text{th}$ estimator as
$\tilde{\theta}_n(k, \eta_0) - \theta_0$, the index of the estimator that
achieves the minimal risk difference is
\begin{equation*}
    \tilde{k}_n \coloneqq \argmin_{k \in \{1,\ldots,K\}} \;
    \tilde{\theta}_n(k, \eta_0) - \theta_0.
\end{equation*}
The subscript $n$ emphasizes that the risk and optimal estimator index are
conditional on the empirical distribution $P_n$. They are therefore random
variables.

Given the high-dimensional nature of the data, it is generally most convenient
to study the performance of estimators of $\boldsymbol{\psi}_0$ using
\textit{Kolmogorov asymptotics}, that is, in the setting in which both $n
\rightarrow \infty$ and $J \rightarrow \infty$ such that $J/n \rightarrow m <
\infty$. Historically, estimators have been derived within this
high-dimensional asymptotic regime to improve upon the finite sample results of
estimators brought about by traditional asymptotic arguments. After all, the
sample covariance matrix retains its asymptotic optimality properties when $J$
is fixed, even though it is known to perform poorly in high-dimensional
settings.

Naturally, it would be desirable for an estimator selection procedure to select
the estimator indexed by $\tilde{k}_n$; however, this quantity depends on the
true, unknown data-generating distribution $P_0$. As a substitute for the
candidates' true conditional risks, we employ instead the cross-validated
estimators of these same conditional risks.

Cross-validation (CV) consists of randomly, and possibly repeatedly,
partitioning a data set into a training set and a validation set. The
observations in the training set are fed to the candidate estimators and the
observations in the validation set are used to evaluate the performance of
these estimators~\citep{breiman1992submodel, friedman2001elements}. A range of
CV schemes have been proposed and investigated, both theoretically and
computationally; \citet{dudoit2005asymptotics} provide a thorough review of
popular CV schemes and their properties. Among the variety, $V$-fold stands out
as an approach that has gained traction on account of its relative
computational feasibility and good performance. Any CV scheme can be expressed
in terms of a binary random vector $B_n$, which assigns observations into
either the training or validation set. Observation $i$ is said to lie in the
training set when $B_n(i) = 0$ and in the validation set otherwise. The
training set therefore contains $\sum_i (1 - B_n(i)) = n (1 - p_n)$
observations and the validation set $\sum_i B_n(i) = n p_n$ observations, where
$p_n$ is the fixed validation set proportion corresponding to the chosen CV
procedure. The empirical distributions of the training and validation sets are
denoted by $P_{n, B_n}^0$ and $P_{n, B_n}^1$, respectively, for any given
realization of $B_n$. \rone{$B_n$, we emphasize, is independent of $P_n$.}

Using this general definition, the cross-validated estimator of a candidate
$\hat{\Psi}_k$'s risk is
\begin{equation*}
  \begin{split}
    \hat{\theta}_{p_n,n}(k, \hat{H}(P_{n,B_n}^0)) & \coloneqq \mathbb{E}_{B_n}\left[
      \Theta(\hat{\Psi}_k(P_{n,B_n}^0), \hat{H}(P_{n,B_n}^0),
      P_{n,B_n}^1)\right] \\
    & =  \mathbb{E}_{B_n} \left[
      \frac{1}{np_n}\sum_{i=1}^n
      \mathbb{I}(B_n(i) = 1) L(X_i; \hat{\Psi}_k(P^0_{n, B_n}),
      \hat{H}(P_{n, B_n}^0))\right],
  \end{split}
\end{equation*}
for a nuisance parameter estimator mapping $\hat{H}$. \rone{Here,
$\mathbb{E}_{B_n}[\cdot]$ denotes the expectation with respect to $B_n$}. The
corresponding cross-validated selector is
\begin{equation*}
  \hat{k}_{p_n,n} \coloneqq \argmin_{k \in \{1, \ldots, K\}}
    \hat{\theta}_{p_n,n}(k, \hat{H}(P_{n,B_n}^0)).
\end{equation*}
 
As a benchmark, the unknown cross-validated conditional risk of the
$k$\textsuperscript{th} estimator is
\begin{equation*}
  \tilde{\theta}_{p_n,n}(k, \eta_0) \coloneqq
    \mathbb{E}_{B_n}\left[
    \Theta(\hat{\Psi}_k(P_{n, B_n}^0), \eta_0, P_0)\right].
\end{equation*}
The cross-validated oracle selector is then
\begin{equation*}
  \tilde{k}_{p_n, n} \coloneqq \argmin_{k \in \{1, \ldots, K\}}
    \tilde{\theta}_{p_n,n}(k, \eta_0).
\end{equation*}
As before, the $p_n$ and $n$ subscripts highlight the dependence of these
objects on the CV procedure and the empirical distribution $P_n$,
respectively, thus making them random variables.

Ideally, the cross-validated estimator selection procedure should identify a
$\hat{k}_{p_n,n}$ that is asymptotically (in $n$, $J$, and possibly $K$)
equivalent in terms of risk to the oracle $\tilde{k}_{p_n,n}$, under a set of
nonrestrictive assumptions based on the choice of loss function, target
parameter space, estimator ranges, and, if applicable, nuisance parameter
space, in the sense that
\begin{equation}\label{intermediate-ratio}
  \frac{\tilde{\theta}_{p_n,n}(\hat{k}_{p_n,n}, \eta_0) - \theta_0}
    {\tilde{\theta}_{p_n,n}(\tilde{k}_{p_n,n}, \eta_0) - \theta_0} \overset{P}{\longrightarrow} 1
  \text{ as } n, J \rightarrow \infty.
\end{equation}
That is, the estimator selected via CV is equivalent in terms of risk to
the CV scheme's oracle estimator chosen from among all candidates.

When Equation~\eqref{intermediate-ratio} holds, a further step may be taken by
relating the performance of the cross-validated selector to that of the
full-dataset oracle selector, $\tilde{k}_n$:
\begin{equation}\label{goal-ratio}
  \frac{\tilde{\theta}_{n}(\hat{k}_{p_n,n}, \eta_0) - \theta_0}
    {\tilde{\theta}_{n}(\tilde{k}_{n}, \eta_0) - \theta_0}
    \overset{P}{\longrightarrow} 1
    \text{ as } n, J \rightarrow \infty.
\end{equation}

When the cross-validated selection procedure's full-dataset conditional risk
difference converges in probability to that of the full-dataset oracle's, the
chosen estimator is \textit{asymptotically optimal}. In other words, the
data-adaptively selected estimator performs asymptotically as well, with respect
to the chosen loss, as the candidate that would be picked from the collection of
estimators if the true data-generating distribution were known.

\section{Loss Functions and Estimator Selection}
\label{losses}

\subsection{Proposed Loss Function}
\label{loss-funs}

\rtwo{The choice of loss function should reflect the goals of the estimation
task.} While loss functions based on the sample covariance matrix and either the
Frobenius or the spectral norms are often employed in the covariance matrix
estimation literature, \citet{dudoit2005asymptotics}'s estimator selection
framework is more amenable to loss functions that operate over random vectors.
Accordingly, we propose the observation-level Frobenius loss:
\begin{equation}\label{observation-level-frob-loss}
  \begin{split}
    L(X; \boldsymbol{\psi}, \eta_0) & \coloneqq \lVert X X^\top -
    \boldsymbol{\psi}
      \rVert^2_{F, \eta_0} \\
    & = \sum_{j=1}^J\sum_{l=1}^J \eta_0^{(jl)} (X^{(j)} X^{(l)} -
      \psi^{(jl)})^2,
  \end{split}
\end{equation}
where $X^{(j)}$ is the $j$\textsuperscript{th} element of a random vector $X
\sim P \in \mathcal{M}$, $\psi^{(jl)}$ is the entry in the
$j$\textsuperscript{th} row and $l$\textsuperscript{th} column of an arbitrary
covariance matrix $\boldsymbol{\psi} \in \mathbf{\Psi}$, and $\eta_0$ is a $J
\times J$ matrix acting as a scaling factor, that is, a potential nuisance
parameter. For an estimator $\hat{\eta}$ of $\eta_0$, the cross-validated
risk estimator of the $k$\textsuperscript{th} candidate estimator
$\hat{\Psi}_k$ under the observation-level Frobenius loss is
\begin{equation*}
  \begin{split}
    \hat{\theta}_{p_n,n}(k, \hat{H}(P_{n, B_n}^0)) &
%       = \mathbb{E}_{B_n} \left[
%       \Theta(\hat{\Psi}_k(P_{n,B_n}^0), \hat{H}(P_{n,B_n}^0),
%       P_{n,B_n}^1)\right] \\
    = \mathbb{E}_{B_n} \left[\frac{1}{np_n} \sum_{i = 1}^n
      \mathbb{I}(B_n(i) = 1)
      \lVert X_i X^\top_i - \hat{\Psi}_k(P_{n,B_n}^0)
      \rVert^2_{F, \hat{H}(P_{n, B_n}^0)}\right].
  \end{split}
\end{equation*}

\citet{ledoit:2004}, \citet{bickelb:2008}, and \citet{rothman:2009}, among
others, have employed analogous (scaled) Frobenius losses to prove various
optimality results, defining $\eta_0^{(jl)} = 1/J$, $\forall j, l$. This
particular choice of scaling factor is such that whatever the value of $J$,
$\lVert \mathbf{I}_{J \times J} \rVert_{F, \eta_0} = 1$. With such a scaling
factor, the loss function may be viewed as a relative loss whose yardstick is
the $J \times J$ identity matrix. A similarly reasonable option for when the
true covariance matrix is assumed to be dense is $\eta_0^{(jl)} = 1/J^2$. This
weighting scheme effectively computes the average squared error across every
entry of the covariance matrix; however, when the scaling factor is constant, it
only impacts the interpretation of the loss. Constant scaling factors have no
impact on our asymptotic analysis. Since it need not be estimated, it is not
a nuisance parameter in the conventional sense.

When the scaling factor of Equation~\eqref{observation-level-frob-loss} is
constant, the risk minimizers are identical for the cross-validated
observation-level Frobenius risk and the common cross-validated Frobenius risk
\rone{
\citep[used by, for example,][]{bickelb:2008,rothman:2009,fan:2013,fang:2015}}.

\begin{proposition}\label{prop:frob}
  Define the cross-validated Frobenius risk for an estimator $\hat{\Psi}_k$ as
  \begin{equation}\label{matrix-frob-risk}
    \hat{R}_n(\hat{\Psi}_k, \eta_0) \coloneqq \mathbb{E}_{B_n}\left[\lVert
      \mathbf{S}_n(P_{n,B_n}^1) - \hat{\Psi}_k(P_{n,B_n}^0)\rVert^2_{F, \eta_0}
   \right],
  \end{equation}
  where $\mathbf{S}_n(P_{n,B_n}^1)$ is the sample covariance matrix computed
  over the validation set $P_{n,B_n}^1$, and $\eta_0$ is some constant scaling
  matrix. \rone{Then, $\hat{R}_n(\hat{\Psi}_k, \eta_0) -
  \hat{\theta}_{p_n,n}(k,\eta_0)$ is constant with respect to
  $\hat{\Psi}_k(P_{n,B_n}^0)$ such that}
  \begin{equation*}
    \begin{split}
      \hat{k}_{p_n,n} & = \argmin_{k \in \{1, \ldots, K\}}
      \hat{\theta}_{p_n,n}(k, \eta_0) \\
                      & = \argmin_{k \in \{1, \ldots, K\}}
                      \hat{R}_n(\hat{\Psi}_k, \eta_0).
    \end{split}
  \end{equation*}
\end{proposition}

Note that the traditional Frobenius loss corresponds to the sum of the squared
eigenvalues of the difference between the sample covariance matrix and the
estimate. Proposition~\ref{prop:frob} therefore implies the existence of a
similar relationship for our observation-level Frobenius loss. It may therefore
serve as a surrogate for a loss based on the spectral norm.

We are not restricted to a constant scaling factor matrix. One might consider
weighting the covariance matrix's off-diagonal elements' errors by their
corresponding diagonal entries, especially useful when the random variables are
of different scales. Such a scaling factor might offer a more equitable
evaluation across all entries of the parameter:
\begin{equation*}
  \begin{split}
      L_{\text{weighted}}(X; \boldsymbol\psi, \eta_0) & \coloneqq
%         \sum_{j=1}^J \sum_{l=1}^J \eta_0^{(jl)}(X^{(j)}X^{(l)}-\psi^{(jl)})^2 \\
      \sum_{j=1}^J \sum_{l=1}^J 
        \frac{1}{\sqrt{\psi_0^{(jj)}\psi_0^{(ll)}}}
        (X^{(j)}X^{(l)}-\psi^{(jl)})^2.
  \end{split}
\end{equation*}
Here, $\eta_0=\text{diag}(\boldsymbol\psi_0)$ is a genuine nuisance parameter
which can be estimated via the diagonal entries of the sample covariance
matrix.

Finally, the covariance matrix $\boldsymbol{\psi}_0$ is the risk minimizer of
the observation-level Frobenius loss if the integral with respect to $X$ and
the partial differential operators with respect to $\boldsymbol\psi$ are
interchangeable.
\rone{
\begin{proposition}\label{prop:minimizer}
  Let the integral with respect to $X$ and the partial differential operators
  with respect to $\boldsymbol\psi$ be interchangeable, and let $\eta$ be some
  fixed $J\times J$ matrix. Then
  \begin{equation*}
    \boldsymbol{\psi}_0 = \argmin_{\boldsymbol{\psi} \in \boldsymbol{\Psi}}\;
    \Theta(\boldsymbol{\psi}, \eta, P_0)
  \end{equation*}
  for $\Theta(\cdot)$ defined under the observation-level Frobenius loss.
\end{proposition}
The proof is provided in the Online Supplement.} Our proposed loss therefore
satisfies the condition of Equation~\eqref{risk-min-def}. The main results of
the paper, however, relate only to the constant scaling factor case. In a minor
abuse of notation, we set $\eta_0 = \boldsymbol{1}$, and suppress dependence of
the loss function on the scaling factor wherever possible throughout the
remainder of the text.

\subsection{Optimality of the Cross-validated Estimator Selector}
\label{theory}

Having defined a suitable loss function, we turn to a discussion of the
theoretical properties of the cross-validated estimator selection procedure.
Specifically, we present, in Theorem~\ref{thm1}, sufficient conditions under
which the method is asymptotically equivalent in terms of risk to the
commensurate CV oracle selector (as per Equation~\eqref{intermediate-ratio}).
This theorem extends the general framework of \citet{dudoit2005asymptotics} for
use in high-dimensional multivariate estimator selection. Adapting their
existing theory to this setting  requires a judicious choice of loss function,
new assumptions, and updated proofs reflecting the use of high-dimensional
asymptotics. Corollary~\ref{full-dataset-result} then builds on
Theorem~\ref{thm1} and details conditions under which the procedure produces
asymptotically optimal selections in the sense of Equation~\eqref{goal-ratio}.
All proofs are provided in Section~\ref{supp:proofs} of the Online
Supplement.

\begin{theorem}\label{thm1}
  
  Let $X_1, \ldots, X_n$ be a random sample of $n$ i.i.d.~random vectors of
  dimension $J$, such that $X_i \sim P_0 \in \mathcal{M}, i = 1, \ldots, n$.
  Assume, without loss of generality, that $\mathbb{E}[X_i] = 0$, and define
  $\boldsymbol{\psi}_0 \coloneqq \text{Var}[X_i]$. Denote the set of $K$
  candidate estimators by $\{\hat{\Psi}_k(\cdot): k = 1, \ldots, K\}$.  Next,
  define the observation-level Frobenius loss function as $L(X;
  \boldsymbol{\psi}) \coloneqq \lVert X^\top X - \boldsymbol{\psi} \rVert^2_{F,
  1}$. Finally, designate $p_n$ as the proportion of observations in any given
  cross-validated validation set.  Consider the following assumptions:
  
  \begin{assumption}\label{thm1:ass1}
    For any $P \in \mathcal{M}$ and $X \sim P$, $\max_{j = 1, \ldots, J}\:
      (\lvert X^{(j)} \rvert) < \sqrt{M_1} < \infty$ almost surely (a.s.).
  \end{assumption}
  \begin{assumption}\label{thm1:ass2}
    \rone{Define $\boldsymbol\Psi \coloneqq \{\boldsymbol\psi \in
    \mathbb{R}^{J\times J} |\; \boldsymbol\psi = \boldsymbol\psi^\top, \lvert
    \psi^{(jl)}\rvert < M_2 < \infty$, $\forall \; j,l = 1, \ldots, J \}$, and
  assume that $\hat{\Psi}_k(P_n), \boldsymbol\psi_0 \in \boldsymbol\Psi$.}
  \end{assumption}

  \noindent
  \textbf{Finite-Sample Result.}
  Let $\overline{M}(J) \coloneqq 4(M_1+M_2)^2J^2$ and $c(\delta,
  \overline{M}(J)) \coloneqq 2(1 + \delta)^2 \overline{M}(J)(1/3 + 1/\delta)$.
  Then, for any $\delta > 0$,
  \begin{equation}\label{thm1:finite-sample-result}
    0 \leq \mathbb{E}_{P_0}[\tilde{\theta}_{p_n,n}(\hat{k}_{p_n,n}) - \theta_0]
      \leq (1 + 2\delta) \mathbb{E}_{P_0}[\tilde{\theta}_{p_n,n}
      (\tilde{k}_{p_n,n})-\theta_0] + 2c(\delta, \overline{M}(J))
      \frac{1 + \log(K)}{np_n}.
  \end{equation}

  \noindent
  \textbf{High-Dimensional Asymptotic Result.}
  The finite-sample result in Equation~\eqref{thm1:finite-sample-result} has the
  following asymptotic implications: If
  $c(\delta, \overline{M}(J))(1 + \log(K))/
  (np_n \mathbb{E}_{P_0}[\tilde{\theta}_{p_n,n}(\tilde{k}_{p_n,n})-\theta_0])
  \rightarrow 0$ and $J/n \rightarrow m < \infty$ as $n, J \rightarrow \infty$,
  then
  \begin{equation}\label{thm1:convergence}
    \frac{\mathbb{E}_{P_0}[\tilde{\theta}_{p_n,n}(\hat{k}_{p_n,n}) - \theta_0]}
      {\mathbb{E}_{P_0}[\tilde{\theta}_{p_n,n}(\tilde{k}_{p_n,n})
      - \theta_0]} \rightarrow 1.
  \end{equation}
  Further, if $c(\delta, \overline{M}(J))(1 + \log(K))/
  (np_n(\tilde{\theta}_{p_n,n}(\tilde{k}_{p_n,n}) -\theta_0))
  \overset{P}{\rightarrow} 0$ as $n, J \rightarrow \infty$, then
  \begin{equation}\label{thm1:convergence-in-prob}
    \frac{\tilde{\theta}_{p_n,n}(\hat{k}_{p_n,n}) - \theta_0}
      {\tilde{\theta}_{p_n,n}(\tilde{k}_{p_n,n}) - \theta_0}
      \overset{P}{\rightarrow} 1.
  \end{equation}
\end{theorem}

The proof relies on special properties of the random variable
$Z_k~\coloneqq~L(X;\hat{\Psi}_k(P_n))~-~L(X; \boldsymbol{\psi}_0)$ and on an
application of Bernstein's inequality~\citep{bennett1962probability}. Together,
they are used to show that $2c(\delta, \overline{M}(J))(1 + \log(K)) / (n p_n)$
is a finite-sample bound for comparing the performance of the cross-validated
selector, $\hat{k}_{p_n,n}$, against that of the oracle selector over the
training sets, $\tilde{k}_{p_n, n}$. Once this bound is established, the
high-dimensional asymptotic results follow immediately.

Only a few sufficient conditions are required to provide finite-sample bounds
on the expected risk difference of the estimator selected via our CV procedure.
First, that each element of the random vector $X$ be bounded, and, second, that
the entries of all covariance matrices in the parameter space and the set of
possible estimates be bounded. Together, these assumptions allow for the
definition of $\overline{M}(J)$, the object permitting the extension of the
loss-based estimation framework to the high-dimensional covariance matrix
estimation problem.

\rone{The first assumption is technical in nature --- it makes the proofs
tractable. While it may appear stringent, and, for instance, is not satisfied
by Gaussian distributions, we believe it to be generally applicable.} We stress
that parametric data-generating distributions, like those exhibiting
Gaussianity, rarely reflect reality, that is, they are merely mathematical
conveniences\footnote{Anecdotally, one cannot help but find themself reminded
that ``Everyone is sure of this [that errors are normally
distributed]~$\ldots$~since the experimentalists believe that it is a
mathematical theorem, and the mathematicians that it is an experimentally
determined fact.''~\citep[][p.~171]{poincare1912calcul}}. Most random variables,
or transformations thereof, arising in scientific practice are bounded by
limitations of the physical, electronic, or biological measurement process;
thus, our method remains widely applicable. For example, in microarray and
next-generation sequencing experiments, the raw data are images on a 16-bit
scale, constraining them to $[0, 2^{16})$. Similarly, the measurement of
immunologic markers, of substantial interest in vaccine efficacy trials of
HIV-1, COVID-19, and other infectious diseases, are bounded by the limits of
detection and/or quantitation imposed by assay biotechnology.

Verifying that the additional assumptions required by Theorem~\ref{thm1}'s
asymptotic results hold proves to be more challenging. We write $f(y) =
O(g(y))$ if $\lvert f \rvert$ is bounded above by $g$, $f(y) = o(g(y))$ if $f$
is dominated by $g$, $f(y) = \Omega(g(y))$ if $f$ is bounded below by $g$, and
$f(y) = \omega(g(y))$ if $f$ dominates $g$, all in asymptotics with respect to
$n$ and $J$. Further, a subscript ``P'' might be added to these bounds, denoting
a convergence in probability. Now, note that $c(\delta, \overline{M}(J))(1
+ \log(K))/(np_n) = O(J)$ for fixed $p_n$ and as $J/n \rightarrow m > 0$. Then
the conditions associated with Equation~\eqref{thm1:convergence} and
Equation~\eqref{thm1:convergence-in-prob} hold so long as
$\mathbb{E}_{P_0}[\tilde{\theta}_{p_n,n}(\tilde{k}_{p_n,n})-\theta_0]
= \omega(J)$ and $\tilde{\theta}_{p_n,n}(\tilde{k}_{p_n,n})-\theta_0
= \omega_P(J)$, respectively.

These requirements do not seem particularly restrictive given that the
complexity of the problem generally increases as a function of the number of
features. There are many more entries in the covariance matrix requiring
estimation than there are observations. This intuition is corroborated by our
extensive simulation study in the following section. Consistent estimation in
terms of the Frobenius risk is therefore not possible in high-dimensions
without additional assumptions about the data-generating process.

Some additional insight might be gained by identifying conditions under which
these assumptions are \textit{unmet} for popular structural beliefs about the
true covariance matrix. In particular, we consider the sparse covariance
matrices defined in~\citet{bickelb:2008} and accompanying hard-thresholding
estimators (see Section~\ref{supp:thresh-est} of the Online Supplement):

\begin{proposition}\label{prop:sparsity}
  In addition to Assumptions~\ref{thm1:ass1} and~\ref{thm1:ass2} of
  Theorem~\ref{thm1}, assume that $\psi_0$ is a member of the following set of
  matrices:
  \begin{equation*}
    \left\{\psi: \psi^{(jj)} < M_2, \sum_{l=1}^J I(\psi^{(jl)}
    \neq 0) < s(J) \text{ for all }, j = 1, \ldots, J \right\}
  \end{equation*}
  where $s(J)$ is the row sparsity, that the hard-thresholding estimator is in
  the library of candidates, \rone{and that it uses a ``sufficiently large''
  thresholding hyperparameter value in the sense of~\citet{bickelb:2008}.}
  Then, by Theorem 2 of~\citet{bickelb:2008}, we have\newline
  $\mathbb{E}_{P_0}[\tilde{\theta}_{p_n,n}(\tilde{k}_{p_n,n})-\theta_0] =
  o(J)$ if $s(J)/J = o(1/\log J)$ asymptotically in $n$ and $J$.
\end{proposition}

Proposition~\ref{prop:sparsity} states that the conditions for achieving the
asymptotic results of Theorem~\ref{thm1} are not met if the proportion of
non-zero elements in the covariance matrix's row with the most non-zero
elements converges to zero faster than $1/\log J$ and the library of candidates
possesses a hard-thresholding estimator whose thresholding hyperparameter is
reasonable in the sense of~\citet{bickelb:2008}'s Theorem 2 and its subsequent
discussion. \rone{Plainly, the true covariance matrix cannot be too sparse if
the collection of considered estimators contains the hard-thresholding
estimator with a correctly specified thresholding hyperparameter value.}

This implies that banded covariance matrices whose number of bands are fixed
for $J$ do not meet the criteria for our theory to apply, assuming that one of
the candidate estimators correctly specifies the number of bands. Nevertheless,
we observe empirically in Section~\ref{simulations} that our cross-validated
procedure selects an optimal estimator when the true covariance matrix is
banded or tapered more quickly in terms of $n$ and $J$ than any other type of
true covariance matrix.

These results are likely explained by the relatively low complexity of the
estimation problem in this setting. High-dimensional asymptotic arguments are
perhaps unnecessary when the proportion of entries needing to be estimated in
the true covariance matrix quickly converges to zero. These limitations of our
theory reflect stringent, and typically unverifiable, structural assumptions
about the estimand. We reiterate that the conditions of Theorem~\ref{thm1} are
generally satisfied. \rone{In situations where the true covariance matrix is
known to possess this level sparsity, practitioners might instead appeal to
Equation~(39) of~\citet{bickelb:2008} to support their use of a
cross-validated estimator selection procedure. This result, coupled with that
of Proposition~\ref{prop:frob}, likely explains the aforementioned simulation
findings of the banded and tapered covariance matrices.}

Now, Theorem~\ref{thm1}'s high-dimensional asymptotic results relate the
performance of the cross-validated selector to that of the oracle selector for
the CV scheme. As indicated by the expression in Equation~\eqref{goal-ratio},
however, we would like our cross-validated procedure to be asymptotically
equivalent to the oracle over the \textit{entire} data set. The conditions to
obtain this desired result are provided in Corollary~\ref{full-dataset-result},
a minor adaptation of previous work by \citet{dudoit2005asymptotics}. This extension
accounts for increasing $J$, thereby permitting its use in high-dimensional
asymptotics.

\begin{corollary}\label{full-dataset-result}

  Building upon Assumptions~\ref{thm1:ass1} and \ref{thm1:ass2} of
  Theorem~\ref{thm1}, we introduce the additional assumptions that, as $n, J
  \rightarrow \infty$ and  $J/n \rightarrow m < \infty$, $p_n \rightarrow 0$,
  $c(\delta, \overline{M}(J))(1 + \log(K))/
  (np_n(\tilde{\theta}_{p_n,n}(\tilde{k}_{p_n,n}) -\theta_0))
  \overset{P}{\rightarrow} 0$, and

  \begin{equation}\label{cor1-condition}
    \frac{\tilde{\theta}_{p_n,n}(\tilde{k}_{p_n,n})-\theta_0}
      {\tilde{\theta}_{n}(\tilde{k}_n)-\theta_0}
      \overset{P}{\rightarrow} 1.
  \end{equation}
  Under these assumptions, it follows that
  \begin{equation}\label{cor1-asymptotic}
    \frac{\tilde{\theta}_{p_n,n}(\hat{k}_{p_n,n}) - \theta_0}
      {\tilde{\theta}_n(\tilde{k}_n) - \theta_0}
      \overset{P}{\longrightarrow} 1.
  \end{equation}
\end{corollary}
The proof is a direct application of the asymptotic results of
Theorem~\ref{thm1}.

As before, the assumption that $c(\delta, \overline{M}(J))(1 + \log(K))/
(np_n(\tilde{\theta}_{p_n,n}(\tilde{k}_{p_n,n}) - \theta_0))
\overset{P}{\rightarrow} 0$ remains difficult to verify, but essentially
requires the estimation error of the oracle to increase quickly as the number
of features grows. That is, $np_n(\tilde{\theta}_{p_n,n}(\tilde{k}_{p_n,n}) -
\theta_0) = \omega_P(J)$. We posit that this condition is generally satisfied,
similarly to the asymptotic results of Theorem~\ref{thm1}.

Now, a sufficient condition for Equation~\eqref{cor1-condition} is that there
exists a $\gamma > 0$ such that
\begin{equation}\label{suff-cond}
  \left(n^\gamma(\tilde{\theta}_{n}(\tilde{k}_n)-\theta_0), \;
    (n(1-p_n))^\gamma(\tilde{\theta}_{p_n,n}(
  \tilde{k}_{p_n,n})-\theta_0)\right) \overset{d}{\rightarrow} (Z, Z),
\end{equation}
for a single random variable $Z$ with $\mathbb{P}(Z > a) = 1$ for some $a > 0$.
For single-split validation, where $\mathbb{P}(B_n = b) = 1$ for some $b \in
\{0, 1\}^n$, it suffices to assume that there exists a $\gamma > 0$ such that
$n^\gamma(\tilde{\theta}_{n}(\tilde{k}_n)-\theta_0) \overset{d}{\rightarrow} Z$
for a random variable $Z$ with $\mathbb{P}(Z > a) = 1$ for some $a > 0$.

Equation~\eqref{cor1-condition} essentially requires that the (appropriately
scaled) distributions of the cross-validated and full-dataset conditional risk
differences of their respective oracle selections converge in distribution as
$p_n \rightarrow 0$. Again, this condition is unrestrictive. As $p_n
\rightarrow 0$, the composition of each training set becomes increasingly
similar to that of the full-dataset. The resulting estimates produced by each
candidate estimator over the training sets and the full-dataset will
correspondingly converge. Naturally, so too will the cross-validated and
full-dataset conditional risk difference distributions of their respective
selections.

While the number of candidates in the estimator library $K$ has been
assumed to be fixed in this discussion of the proposed method's asymptotic
results, it may be allowed to grow as a function of $n$ and $J$. Of course,
this will negatively impact the convergence rates of $c(\delta,
\overline{M}(J))(1 + \log(K))/(np_n
\mathbb{E}_{P_0}[\tilde{\theta}_{p_n,n}(\tilde{k}_{p_n,n})-\theta_0])$ and
$c(\delta, \overline{M}(J))(1 + \log(K))/
(np_n(\tilde{\theta}_{p_n,n}(\tilde{k}_{p_n,n}) -\theta_0))$. The sufficient
conditions outlined in the asymptotic results of Theorem~\ref{thm1} are
achieved so long as the library of candidates does not grow too aggressively.
That is, we can make the additional assumptions that $K = o(\text{exp}\{
\mathbb{E}_{P_0}[\tilde{\theta}_{p_n,n}(\tilde{k}_{p_n,n})-\theta_0]/J\})$ and
$K = o_P(\text{exp}\{(\tilde{\theta}_{p_n,n}(\tilde{k}_{p_n,n})-\theta_0)/J\})$
such that the results of Equations~\eqref{thm1:convergence} and
\eqref{thm1:convergence-in-prob} are achieved, respectively.

\rone{
Finally, we have assumed thus far that $\mathbb{E}_{P_0}[X]$ is known. This is
generally not the case in practice. In place of a random vector centered at
zero, we might instead consider the set of $n$ demeaned random vectors
$\tilde{\mathbf{X}}_{n\times J}$ where $\tilde{X}_i = X_i - \bar{X}$ and
$\bar{X}^{(j)} = 1/n \sum X_i^{(j)}$. It follows from the details given in
Remark~\ref{supp:remark:unknown-mean} of the Online Supplement that the
asymptotic results of Theorem~\ref{thm1} and
Corollary~\ref{full-dataset-result} apply to $\tilde{\mathbf{X}}_{n\times J}$.
}

\section{Simulation Study}
\label{simulations}

\subsection{Simulation Study Design}

We conducted a series of simulation experiments using prominent covariance
models to verify the theoretical results of our cross-validated estimator
selection procedure. These models are described below.

\begin{description}
  \item[Model 1:] A dense covariance matrix, where
    \begin{equation*}
      \psi^{(jl)} = \begin{cases}
        1, & j = l \\
        0.5, & \text{otherwise} \\
      \end{cases}.
    \end{equation*}
  
  \item[Model 2:] An AR(1) model, where $\psi^{(jl)} = 0.7^{\lvert j-l
    \rvert}$. This covariance matrix, corresponding to a common timeseries
    model, is approximately sparse for large $J$, since the off-diagonal
    elements quickly shrink to zero. 

  \item[Model 3:] An MA(1) model, where $\psi^{(jl)} = 0.7^{\lvert j-l \rvert}
    \cdot \mathbb{I}(\lvert j-l \rvert \leq 1)$. This
    covariance matrix, corresponding to another common timeseries model, is
    truly sparse. Only the diagonal, subdiagonal, and superdiagonal contain
    non-zero elements. 

  \item[Model 4:] An MA(2) model, where
    \begin{equation*}
      \psi^{(jl)} = \begin{cases}
        1, & j = l \\
        0.6, & \lvert j - l \rvert = 1 \\
        0.3, & \lvert j - l \rvert = 2 \\
        0, & \text{otherwise} \end{cases}.
    \end{equation*}
    This timeseries model is similar to Model 3, but slightly less sparse.
  
  \item[Model 5:] A random covariance matrix model. First, a $J \times J$
    random matrix whose elements are i.i.d.~$\text{Uniform}(0, 1)$ is
    generated. Next, entries below $1/4$ are set to $1$, entries between $1/4$
    and $1/2$ are set to $-1$, and the remaining entries are set to $0$. The
    square of this matrix is then computed and added to the identity matrix
    $\mathbf{I}_{J \times J}$. Finally, the corresponding correlation matrix is
    computed and used as the model's covariance matrix.
  
  \item[Model 6:] A Toeplitz covariance matrix, where
    \begin{equation*}
      \psi^{(jl)} = \begin{cases}
        1, & j = l \\
        0.6 \lvert j - l \rvert^{-1.3}, & \text{otherwise}
      \end{cases}.
    \end{equation*}
    Like the AR(1) model, this covariance matrix is approximately sparse for
    large $J$. However, the off-diagonal entries decay less quickly as their
    distance from the diagonal increases.

  \item[Model 7:] A Toeplitz covariance matrix with alternating signs, where
    \begin{equation*}
      \psi^{(jl)} = \begin{cases}
        1, & j = l \\
        (-1)^{\lvert j-l \rvert}0.6\lvert j-l \rvert^{-1.3}, & \text{otherwise}
      \end{cases}.
    \end{equation*}
    This model is almost identical to Model 6, though the signs of the
    covariance matrix's entries are alternating.

  \item[Model 8:] A covariance matrix inspired by the latent variable model
    described in the Online Supplement's
    Equation~\eqref{supp:factor-model}. Let $\boldsymbol{\beta}_{J \times
    3}~=~(\beta_1, \ldots, \beta_J)^\top$, where $\beta_j$ are randomly
    generated using a $N(0, \mathbf{I}_{3 \times 3})$ distribution for $j = 1,
    \ldots, J$. Then $\boldsymbol\psi = \boldsymbol{\beta}
    \boldsymbol{\beta}^\top + \mathbf{I}_{J \times J}$ is the covariance matrix
    of a model with three latent factors.

\end{description}

Each covariance model was used to generate data sets consisting of $n \in \{50,
100, 200, 500\}$ i.i.d.~multivariate Gaussian, mean-zero observations. The
uniform boundedness condition of \rone{Theorem~\ref{thm1}'s
Assumption~\ref{thm1:ass1}} is therefore not satisfied; we do this purposefully
to further stress that this assumption is not limiting in many practical
settings. For each model and sample size, five data dimension ratios were
considered: $J/n \in \{0.3, 0.5, 1, 2, 5\}$.  Together, the eight covariance
models, four sample sizes, and five dimensionality ratios result in $160$
distinct simulation settings. For each such setting, the performance of the
cross-validated selector with respect to the various oracle selectors and
several well-established estimators is evaluated based on aggregation across
$200$ Monte Carlo repetitions.

We applied our estimator selection procedure, which we refer to as
\textit{cvCovEst}, using a 5-fold CV scheme. The library of candidate
estimators is provided in the Online Supplement's
Table~\ref{supp:table:sim-hyperparams}, which includes details on these
estimators' possible hyperparameters. Seventy-four estimators make up the
library of candidates. We note that no penalty is attributed to estimators
generating rank-deficient estimates, like the sample covariance matrix when $J
> n$, though this limitation is generally of practical importance. When the
situation dictates that the resulting estimate must be positive-definite, the
library of candidates should be assembled accordingly.

\subsection{Simulation Study Results}

To examine empirically the optimality results of Theorem \ref{thm1}, we
computed analytically, for each replication, the cross-validated conditional
risk differences of the cross-validated selection, $\hat{k}_{p_n,n}$, and the
cross-validated oracle selection, $\tilde{k}_{p_n, n}$. The Monte Carlo
expectations of the risk differences stratified by $n$, $J/n$, and the
covariance model were computed from the cross-validated conditional risk
differences. The ratios of the expected risk differences are presented in
Figure~\ref{sim:conv-in-exp}A. These results make clear that, for virtually all
models considered, the estimator chosen by our proposed cross-validated
selection procedure has a risk difference asymptotically identical on average
to that of the cross-validated oracle.

\begin{figure}
  \centering
  \includegraphics[width=\textwidth]{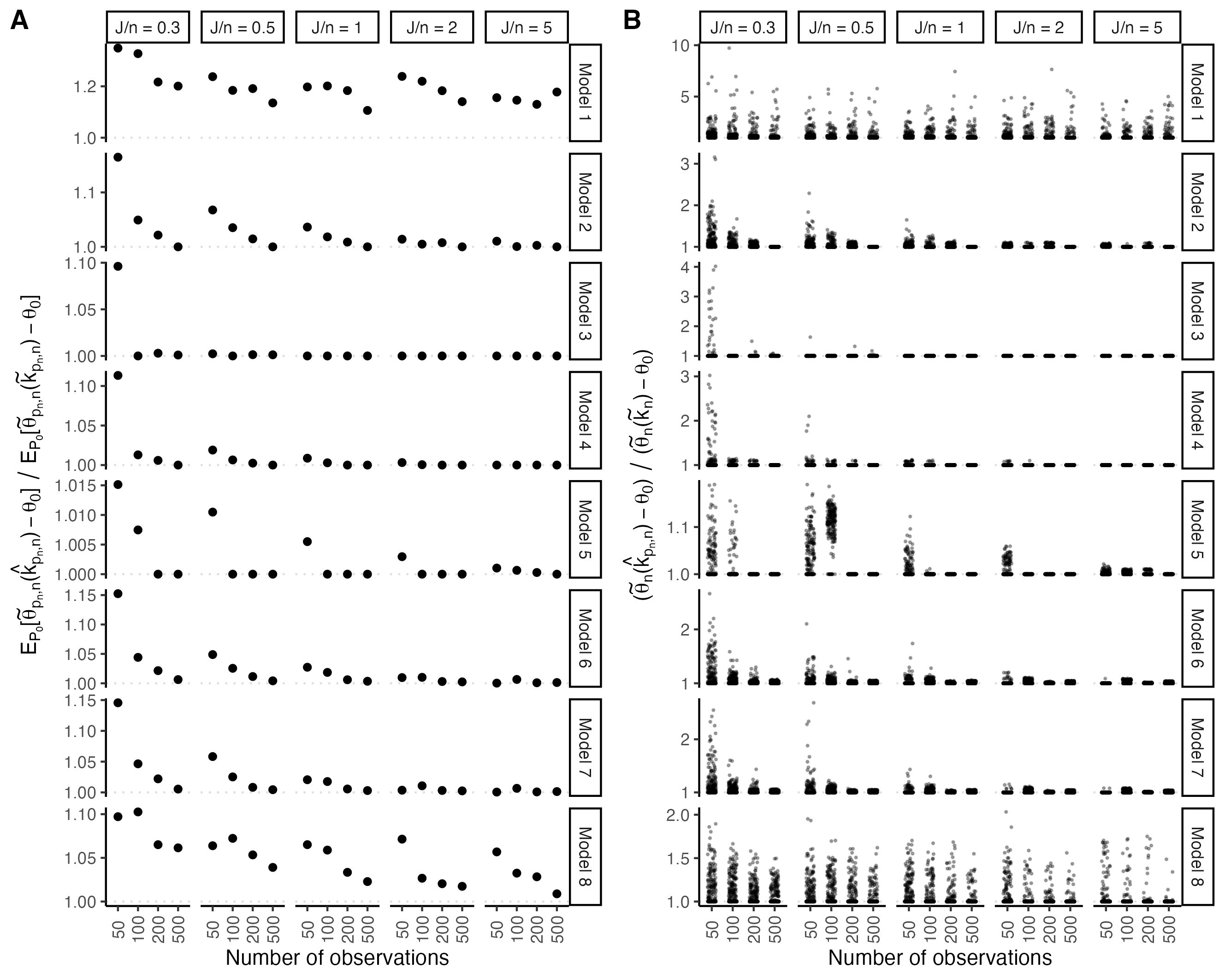}
  \caption{(\textbf{A}) Comparison of the cross-validated selection and
    cross-validated oracle selection's mean cross-validated conditional risk
    differences. (\textbf{B}) Comparison of the cross-validated selection and
    oracle selection's full-dataset conditional risk differences. Note the
    differing y-axis scales for the different models.}
  \label{sim:conv-in-exp}
\end{figure}

A stronger result, corresponding to Equation~\eqref{thm1:convergence-in-prob}
of Theorem~\ref{thm1}, is presented in the Online Supplement's
Figure~\ref{supp:sim:conv-in-prob}. For all but Models 1 and 8, we find that
our algorithm's selection is virtually equivalent to the cross-validated oracle
selection for $n \geq 200$ and $J/n \geq 0.5$. Even for Model 8, in which the
covariance matrices are more difficult to estimate due to their dense
structures, we find that our selector identifies the optimal estimator with
probability tending to $1$ for $n \geq 200$ and $J/n = 5$.

More impressive still are the results presented in
Figure~\ref{sim:conv-in-exp}B that characterize the full-dataset conditional
risk difference ratios. For all covariance matrix models considered, with the
exception of Model 1, our procedure's selections attain near asymptotic
optimality for moderate values of $n$ and $J/n$. This suggests that our
loss-based estimator selection approach's theoretical guarantee, as outlined in
Corollary~\ref{full-dataset-result}, is achievable in many practical settings.

In addition to verifying our method's asymptotic behavior, we compared its
estimates against those of competing approaches. We computed the Frobenius and
spectral norms of each procedure's estimate against the corresponding true
covariance matrix. The mean norms over all simulations were then computed for
each covariance matrix estimation procedure, again stratified by $n$, $J/n$,
and the covariance matrix model (Figures~\ref{supp:mean-frobenius} and
\ref{supp:mean-spectral} of the Online Supplement). Our data-adaptive
procedure is found to perform at least as well as the best alternative
estimation strategy across all simulation scenarios under both norms.

\section{Real Data Examples}
\label{examples}

Single-cell transcriptome sequencing (scRNA-seq) allows researchers to study
the gene expression profiles of individual cells. The fine-grained
transcriptomic data that it provides have been used to identify rare cell
populations and to elucidate the developmental relationships between diverse
cellular states.

\begin{figure}
  \centering
  \includegraphics[width=0.7\textwidth]{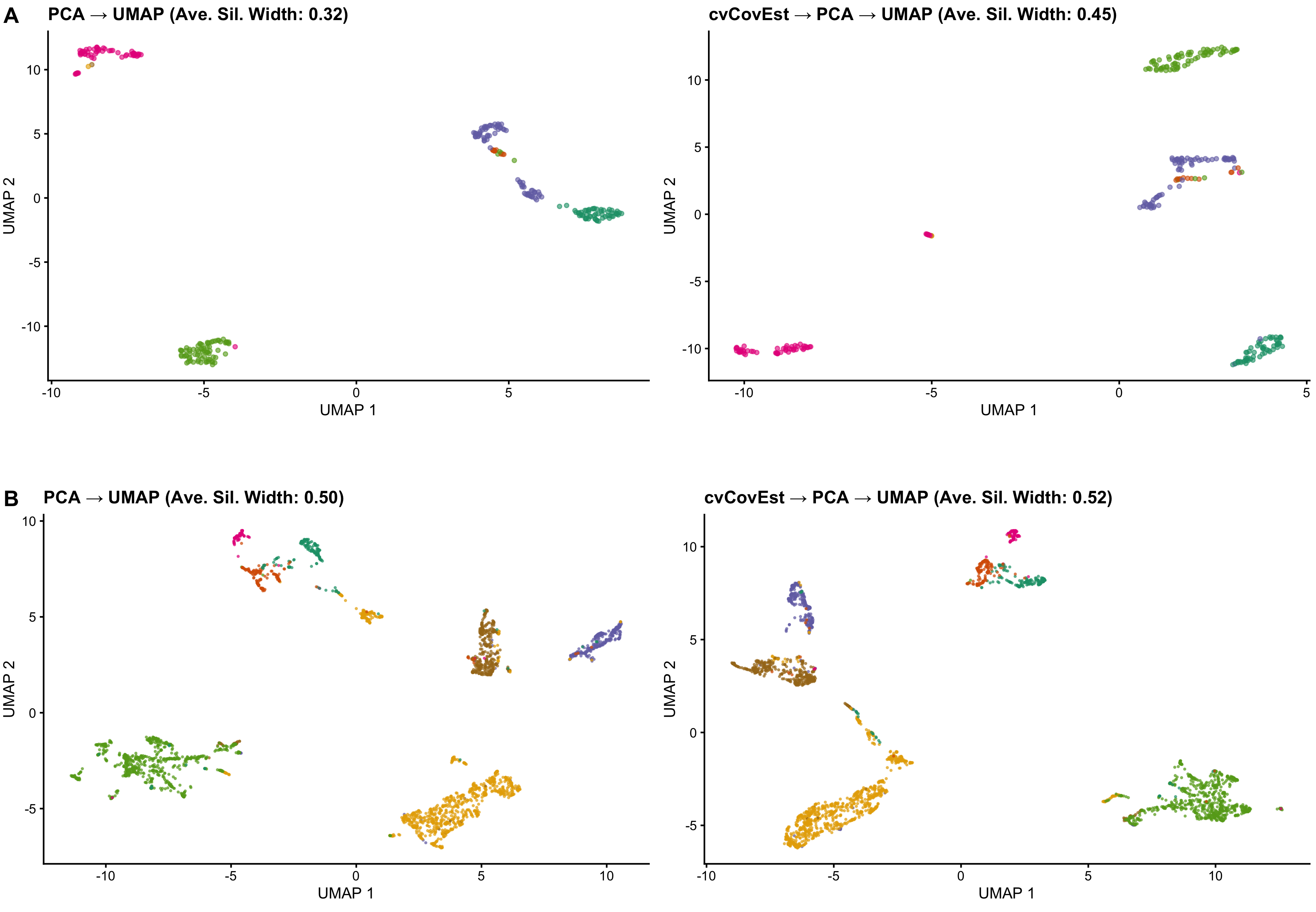}
  \caption{
    Comparisons of scRNA-seq data sets' UMAP embeddings based on vanilla PCA or
    PCA with the cross-validated selection's covariance matrix estimate. The
    data sets consist of \textbf{(A)} 285 cells collected from the visual
    cortex of mice and \textbf{(B)} 2,816 mouse brain cells. Distinct cell types
    are indicated by color.
  }
  \label{scRNA-seq}
\end{figure}

Given that a typical scRNA-seq data set possesses tens of thousands of features
(genes), most workflows prescribe a dimensionality reduction step. In addition
to reducing the amount of computational resources needed to analyze the data,
reducing the dimensions mitigates the effect of corrupting noise on interesting
biological signal. The lower-dimensional embedding is then used in downstream
analyses, like novel cell-type discovery via clustering.

One of the most popular methods used to reduce the dimensionality of scRNA-seq
data is uniform manifold approximation and projection
(UMAP)~\citep{McInnes:2018}. This method captures the most salient non-linear
relationships among a high-dimensional data set's features and projects them
onto a reduced-dimensional space. Instead of applying UMAP directly, the
scRNA-seq data set's leading principal components (PCs) are often used as an
initialization.

This initial dimensionality reduction by PCA is believed to play a helpful role
in denoising. However, PCA typically relies on the sample covariance matrix,
and so when the data set is high-dimensional, the resulting principal
components are known to be poor estimates of those of the
population~\citep{Johnstone:2009}. We hence posit that our cross-validated
estimator selection procedure could form a basis for an improved PCA. That is,
we hope that the eigenvectors resulting from the eigendecomposition of our
estimated covariance matrix could be used to generate a set of estimates closer
to the true PCs in terms of risk. These PCs could then be fed to UMAP to
produce an enhanced embedding. Indeed, additional simulation results provided
in Section~\ref{supp:spectral-norm-sim} of the Online Supplement suggest
that cvCovEst produces estimates of the leading eigenvalue at least as well as
those produced by the sample covariance matrix, in terms of the spectral norm.

We applied our procedure to two scRNA-seq data sets for which the cell types
are known \textit{a priori}. These data were obtained from the
\texttt{scRNAseq} Bioconductor \texttt{R} package~\citep{scRNAseq-pkg}, and
prepared for analysis using a workflow outlined in \citet{osca:2020}. A 5-fold
CV scheme was used; the library of candidate estimators is provided in the
Online Supplement's Table~\ref{supp:table:scrnaseq-hyperparams}. We expect that
cells of the same class will form tight, distinct clusters within the
low-dimensional representations. The resulting embeddings, which we refer to as
the \textit{cvCovEst-based} embeddings, were then compared to those produced by
UMAP using traditional PCA for initialization, which we refer to as the
\textit{PCA-based} embeddings. For each embedding, the 20 leading PCs were fed
to UMAP. The first data set is a collection of 285 mouse visual cortex cells
\citep{tasic:2016}, and the second data set consists of 2,816 mouse brain cells
\citep{zeisel:2015}. The 1,000 most variable genes of each data set were used
to compute the PCs of both embeddings. 

\begin{figure}
  \centering
  \includegraphics[width=0.8\textwidth]{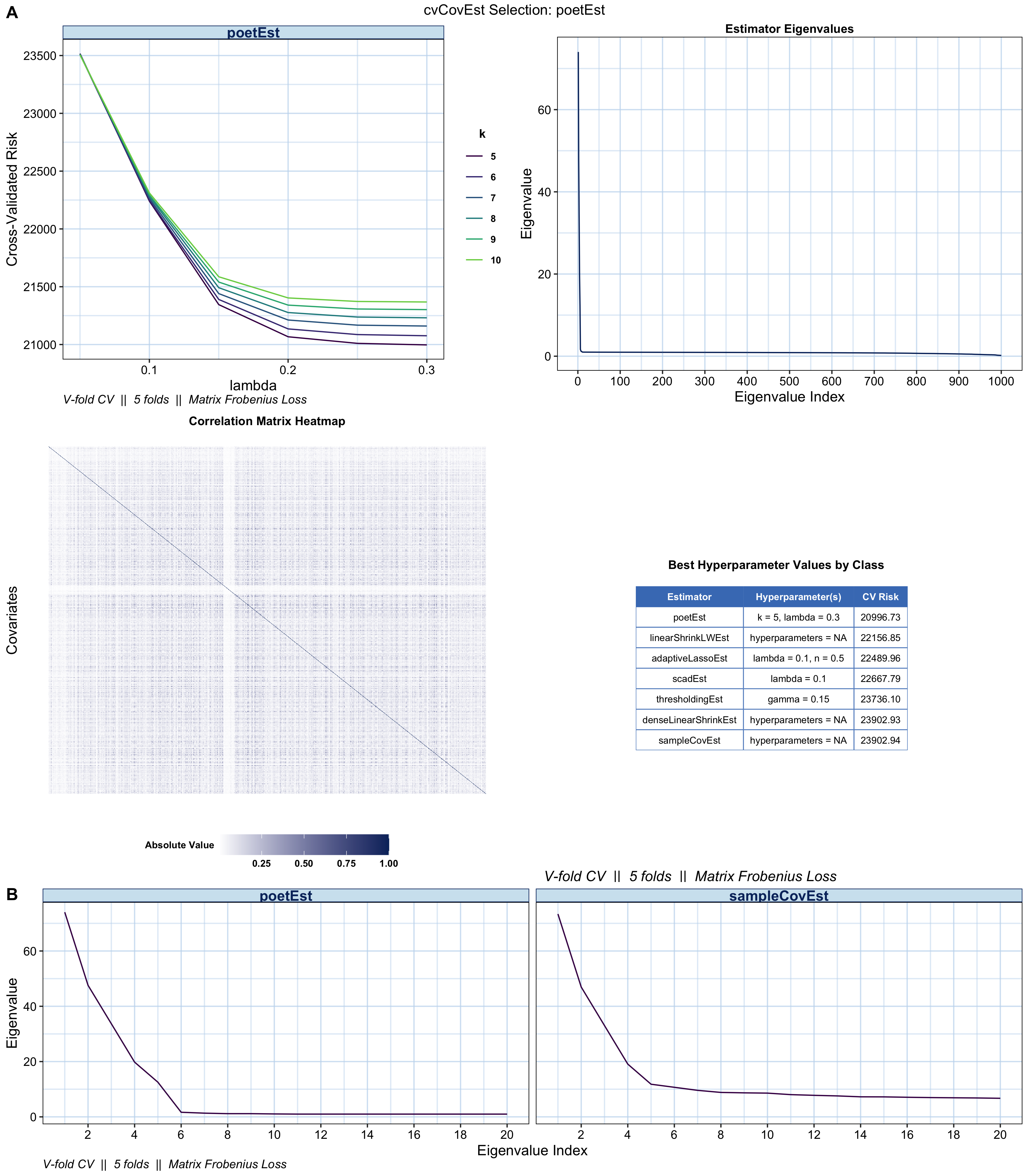}
  \caption{Tasic dataset: Diagnostic plots and tables generated 
  using the cvCovEst R package.
  \textbf{(A)} The top-left plot presents the cross-validated Frobenius risk of
  the estimator selected by our method. $k$ represents the number of
  potential latent factors, and \texttt{lambda} the thresholding value used.
  The top-right panel contains a line plot of the selected estimator's
  eigenvalues. The bottom-left plot displays the absolute values of the
  estimated correlation matrix output by the cvCovEst selection, and the
  bottom-right table lists the best performing estimators from all classes of
  estimators considered. \textbf{(B)} Side-by-side line plots of the estimated
  leading eigenvalues of the cvCovEst selection and the sample covariance
  matrix.}
  \label{fig:allen-diagnostics}
\end{figure}

The resulting UMAP plots are presented in Figure~\ref{scRNA-seq}. Though the
two embeddings generated for each data set are qualitatively similar, the
low-dimensional representation relying on our loss-based approach is more
refined in Figure~\ref{scRNA-seq}A. A number of cells erroneously clustered in
the PCA-based embedding are correctly represented in the cvCovEst-based
embedding. This explains the 41\% increase in average silhouette width of our
method relative to the traditional approach. \rone{Further insight is gleaned
from the diagnostic plots of Figure~\ref{fig:allen-diagnostics}.
Figure~\ref{fig:allen-diagnostics}A indicates that cvCovEst selected the POET
estimator \citep{fan:2013} with 5 latent factors and a thresholding
hyperparameter of 0.3. It that the selected estimator significantly improves
upon the sample covariance matrix in terms of the cross-validated Frobenius
risk. Figure~\ref{fig:allen-diagnostics}B provides further insight into the
discrepancies between the UMAP results of Figure~\ref{scRNA-seq}A: the sample
covariance matrix likely over-estimates many of the leading eigenvalues.}

The embeddings in Figure~\ref{scRNA-seq}B qualitatively identical, and so too
are their average silhouette widths. This is expected, the \citet{zeisel:2015}
is not truly high-dimensional. The sample covariance matrix likely is a
reasonable estimator in this setting. Ideally, data-adaptive selection
procedures should be cognizant of this. Indeed, cvCovEst, when applied to the
\citet{zeisel:2015} data set, selects an estimator whose cross-validated
empirical risk is only slightly smaller than that of the sample covariance
matrix, and whose leading PCs are virtually identical
(Figure~\ref{supp:fig:zeisel-diagnostics} of the Online Supplement).

\section{Discussion}
\label{discussion}

This work extends \citet{dudoit2005asymptotics}'s framework for asymptotically
optimal, data-adaptive estimator selection to the problem of covariance matrix
estimation in high-dimensional settings. We provide sufficient conditions under
which our cross-validated procedure is asymptotically optimal in terms of risk,
and show that it generalizes the cross-validated hyperparameter selection
procedures employed by existing estimation approaches. \rtwo{Future work might
derive analogous results for other loss functions, or perhaps even for other
parameters like the precision matrix.}

The simulation study provides evidence that near-optimal results are achieved
in data sets with relatively modest numbers of observations and many features
across models indexed by diverse covariance matrix structures. These results
also establish that our cross-validated procedure performs as well as the best
bespoke estimation procedure in a variety of settings. Our scRNA-seq data
examples further illustrate the utility of our approach in fields where
high-dimensional data are collected routinely.

Practitioners need no longer rely upon intuition alone when deciding which
candidate estimator is best from among a library of diverse estimators. We
expect that a variety of computational procedures relying upon the accurate
estimation of the covariance matrix beyond the exploratory analyses considered
here, like clustering and latent variable estimation, stand to benefit from the
application of this framework.

\bigskip
\begin{center}
{\large\bf SUPPLEMENTARY MATERIAL}
\end{center}

\begin{description}

\item[Online Supplement:] All technical proofs and additional simulation results
  are included in the online supplement. A brief review of popular covariance
  matrix estimators is also included.

\item[cvCovEst R package:] The cross-validated covariance matrix
  estimator selection procedure is implemented in \texttt{cvCovEst}
  \citep{cvCovEst}, an open-source \texttt{R} package \citep{Rlang}. This
  package is available via the Comprehensive \texttt{R} Archive Network (CRAN)
  at \url{https://CRAN.R-project.org/package=cvCovEst}. Internally, the
  \texttt{cvCovEst} package relies upon the generalized CV framework of the
  \texttt{origami} \texttt{R} package \citep{Coyle:2018}.

\item[Simulation and Analysis Code:] Results in Sections~\ref{simulations}
  and~\ref{examples} were produced using version 0.1.3 of the \texttt{cvCovEst}
  \texttt{R} package. The code and data used to produce these analyses are
  publicly available on GitHub at
  \url{https://github.com/PhilBoileau/pub_cvCovEst}. 

\end{description}

\subsection*{Acknowledgements}
We thank Brian Collica and Jamarcus Liu for significant contributions to the
development of the \texttt{cvCovEst} \texttt{R} package. PB is funded by the
Fonds de recherche du Qu\'{e}bec - Nature et technologies, and the Natural
Sciences and Engineering Research Council of Canada. NSH gratefully acknowledges
the support of the National Science Foundation (award no.~DMS 2102840).

\bibliographystyle{jasa3}
{\footnotesize \bibliography{Bibliography-MM-MC}}

\end{document}

\def\spacingset#1{\renewcommand{\baselinestretch}%
{#1}\small\normalsize} \spacingset{1}

%%%%%%%%%%%%%%%%%%%%%%%%%%%%%%%%%%%%%%%%%%%%%%%%%%%%%%%%%%%%%%%%%%%%%%%%%%%%%%

\if0\blind
{
  \title{\bf Supplementary Material for ``Cross-Validated Loss-Based Covariance
    Matrix Estimator Selection in High Dimensions''}
\author{Philippe Boileau\\
    {\normalsize Graduate Group in Biostatistics and Center for Computational
      Biology,} {\normalsize UC Berkeley} \\
    Nima S.\ Hejazi\\
    {\normalsize Division of Biostatistics, Department of Population Health
        Sciences,} {\normalsize Weill Cornell Medicine}\\
    Mark J.~van der Laan\\
    {\normalsize Division of Biostatistics, Department of Statistics,}\\
    {\normalsize and Center for Computational Biology, UC Berkeley}\\
    Sandrine Dudoit\thanks{To whom correspondence should be addressed:
    sandrine@stat.berkeley.edu}\hspace{0.2cm}\\
    {\normalsize Department of Statistics, Division of Biostatistics,}\\
    {\normalsize and Center for Computational Biology, UC Berkeley}
}
  \maketitle
} \fi

\if1\blind
{
  \bigskip
  \bigskip
  \bigskip
  \begin{center}
    {\LARGE\bf Cross-Validated Loss-Based Covariance Matrix Estimator Selection
      in High Dimensions}
\end{center}
  \medskip
} \fi

\makeatletter\@input{get-main-aux.tex}\makeatother

\newpage
\section{Proofs}
\label{proofs}

\begin{proof}
  \textbf{Proposition~\ref{paper:prop:frob}.}

 Assume without loss of generality that $\mathbb{E}_{P_0}[X] = 0$ and
 $\eta_0=1$. Then,

  \begin{equation*}
    \begin{split}
      \hat{\theta}_{p_n,n}(k, 1) &
        = \mathbb{E}_{B_n}\left[\frac{1}{np_n}
        \sum_{i=1}^n \mathbb{I}(B_n(i) = 1) \lVert
        X_i X^\top_i - \hat{\Psi}_k(P_{n,B_n}^0)
        \rVert^2_{F, 1}\right] \\
      & = \mathbb{E}_{B_n}\left[\frac{1}{np_n}
        \sum_{i=1}^n \mathbb{I}(B_n(i) = 1)
        \sum_{j=1}^J\sum_{l=1}^J (X_{i}^{(j)}X^{(l)}_i -
        \hat{\Psi}_k(P_{n, B_n}^0)^{(jl)})^2 \right]\\
      & = \mathbb{E}_{B_n}\left[\frac{1}{np_n}
        \sum_{j=1}^J\sum_{l=1}^J \left(\sum_{\{i: B_n(i)=1\}}\left(
        (X_i^{(j)}X^{(l)}_i)^2 - 2X_i^{(j)}X^
        {(l)}_i\hat{\Psi}_k(P_{n, B_n}^0)^{(jl)} \right) \right.\right.\\
      & \qquad\qquad\qquad\qquad\qquad\qquad
        + (\hat{\Psi}_k(P_{n, B_n}^0)^{(jl)})^2\Bigg)\Bigg]\\
      & = \mathbb{E}_{B_n} \Bigg[
        \sum_{j=1}^J\sum_{l=1}^J \bigg((\hat{\Psi}_k(P_{n, B_n}^0)^{(jl)})^2 -
        2S(P_{n,B_n}^{1})^{(jl)}\hat{\Psi}_{k}(P_{n, B_n}^0)^{(jl)} \\
        & \qquad \qquad \qquad \qquad \qquad
        + \frac{1}{np_n}\sum_{\{i: B_n(i)=1\}}(X_i^{(j)}X^{(l)}_i)^2
        \bigg)\Bigg] \\
      & = \mathbb{E}_{B_n} \left[
        \sum_{j=1}^J\sum_{l=1}^J \bigg(\big(\hat{\Psi}_k(P_{n, B_n}^0)^{(jl)}
        \big)^2 - 2S(P_{n,B_n}^{1})^{(jl)}\hat{\Psi}_{k}(P_{n, B_n}^0)^{(jl)}
        \bigg)\right] + C_1,
    \end{split}
  \end{equation*}
  where $C_1$ is constant with respect to $\hat{\Psi}_k(P_{n,B_n}^0)$.

  From Equation~\eqref{paper:matrix-frob-risk} of the main text, notice that
  \begin{equation*}
    \begin{split}
      \hat{R}_n(\hat{\Psi}_k, 1) & = \mathbb{E}_{B_n} \left[\lVert
        \hat{\Psi}_k(P_{n,B_n}^0) -
        \mathbf{S}_n(P_{n,B_n}^1)\rVert^2_{F, 1}\right] \\
      & = \mathbb{E}_{B_n}\left[\sum_{j=1}^J\sum_{l=1}^J\left(
        (\hat{\Psi}_k(P_{n, B_n}^0)^{(jl)})^2 -
        2S(P_{n,B_n}^{1})^{(jl)}\hat{\Psi}_{k}(P_{n, B_n}^0)^{(jl)} + 
        (S(P_{n,B_n}^{1})^{(jl)})^2\right)\right] \\
      & = \mathbb{E}_{B_n}\left[\sum_{j=1}^J\sum_{l=1}^J\left(
        (\hat{\Psi}_k(P_{n, B_n}^0)^{(jl)})^2 -
        2S(P_{n,B_n}^{1})^{(jl)}
        \hat{\Psi}_{k}(P_{n, B_n}^0)^{(jl)}\right)\right] + C_2,
    \end{split}
  \end{equation*}
  where $C_2$ is constant with respect to
  $\hat{\Psi}_k(P_{n,B_n}^0)$.
  
  \noindent Thus,
  \begin{equation*}
    \begin{split}
      \hat{k}_{p_n,n} & = \text{arg min}_{k \in \{1, \ldots, K\}}
        \hat{\theta}_{p_n,n}(k, 1) \\
      & = \text{arg min}_{k \in \{1, \ldots, K\}}
        \hat{R}_n(\hat{\Psi}_k, 1).
    \end{split}
  \end{equation*}
\end{proof}

\begin{proof}
  \textbf{Proposition~\ref{paper:prop:minimizer}.}

  For any $a = 1, \ldots, J$ and $b = 1,\ldots, J$, we find that
  \begin{equation*}
    \begin{split}
      0
      & = \frac{\delta}{\delta\psi^{(ab)}}\Theta(\boldsymbol{\psi},
        \boldsymbol{\eta}, P_0) \\
      & = \frac{\delta}{\delta\psi^{(ab)}}\mathbb{E}_{P_0}\left[
        L(X; \boldsymbol{\psi}, \boldsymbol{\eta})\right] \\
      & = \frac{\delta}{\delta\psi^{(ab)}}\mathbb{E}_{P_0}\left[
        \sum_{j=1}^J\sum_{l=1}^J \eta^{(jl)}
        \left(X^{(j)}X^{(l)} - \psi^{(jl)}\right)^2\right] \\
      & = \mathbb{E}_{P_0}\left[\frac{\delta}{\delta\psi^{(ab)}}
        \sum_{j=1}^J\sum_{l=1}^J \eta^{(jl)}
        \left(X^{(j)}X^{(l)} - \psi^{(jl)}\right)^2
      \right] \\
      & \propto \mathbb{E}_{P_0}\left[X^{(a)}X^{(b)} - \psi^{(ab)}\right] \\
      \Rightarrow \psi^{(ab)} & = \psi_0^{(ab)}.
    \end{split}
  \end{equation*}
  It follows that $\boldsymbol{\psi}_0$ is a risk minimizer.
\end{proof}

\begin{proof}
  \textbf{Proposition~\ref{paper:prop:sparsity}.}

  Let the hard-thresholding estimator minimizing the expected cross-validated
  conditional risk under $P_0$ be indexed by $k_0$. Then
  \begin{equation*}
    \begin{split}
      \mathbb{E}_{P_0}[\tilde{\theta}_{p_n,n}(\tilde{k}_{p_n,n})-\theta_0]
      & \leq \mathbb{E}_{P_0}[\tilde{\theta}_{p_n,n}(k_0)-\theta_0] \\
      & = \mathbb{E}_{P_0}\left[\mathbb{E}_{B_n}\left[\mathbb{E}_{P_0}\left[
        \lVert \hat{\Psi}_{k_0}(P_{n, B_n}^0) - XX^\top\rVert^2_{F,1}
        \right.\right.\right.\\
        & \qquad\qquad\qquad\qquad\qquad
        - \lVert \boldsymbol\psi_0 - XX^\top\rVert^2_{F,1}
        \mid P_{n, B_n}^0\Big]\Big]\Big] \\
      & = \mathbb{E}_{P_0}\left[\mathbb{E}_{B_n}\left[\mathbb{E}_{P_0}\left[
        \lVert \hat{\Psi}_{k_0}(P_{n, B_n}^0) - \boldsymbol\psi_0 \rVert^2_{F,1}
        \mid P_{n, B_n}^0\right]\right]\right] \\
      & = O(s(J) \log J).
    \end{split}
  \end{equation*}
  The first inequality follows from the definition of the cross-validated
  conditional risk under $P_0$; the last equality follows from Theorem~2
  (and its subsequent discussion) of \citet{bickelb:2008}, since $X$ is
  sub-Gaussian by the boundedness condition of Assumption~1 and
  $J/n \rightarrow m$ as $n, J \rightarrow \infty$. Then,
  \begin{equation*}
    \frac{c(\delta, \overline{M}(J))(1 + \log(K)}
    {np_n \mathbb{E}_{P_0}[\tilde{\theta}_{p_n,n}(\tilde{k}_{p_n,n})-\theta_0]}
    = \Omega\left(\frac{J}{s(J)\log J}\right).
  \end{equation*}
\end{proof}

\begin{lemma}\label{lma1:bounded-loss-diff}
  Let $Z_k \coloneqq L(X; \hat{\Psi}_k(P^0_{n, B_n})) -
  L(X; \boldsymbol{\psi}_0)$. Then,
  \begin{equation*}
    \text{Var}_{P_0}\left[Z_k \mid P^0_{n, B_n}, B_n\right] \leq
    \overline{M}(J) \mathbb{E}_{P_0}\left[Z_k \mid P^0_{n, B_n}, B_n\right],
  \end{equation*}
  where $\overline{M}(J) \coloneqq 4(M_1+M_2)^2J^2 $.
\end{lemma}
\begin{proof}
  \begin{equation*}
    \begin{split}
      \mathbb{E}_{P_0}\left[Z_k \mid P^0_{n, B_n}, B_n\right] & =
        \mathbb{E}_{P_0}\left[L(X; \hat{\Psi}_k(P^0_{n, B_n})) -
          L(X; \boldsymbol{\psi}_0) \Big\vert
        P^0_{n, B_n}, B_n\right] \\
      & =  \mathbb{E}_{P_0}\left[\sum_{j=1}^J\sum_{l=1}^J
        (X^{(j)}X^{(l)} -
        \hat{\Psi}^{(jl)}_k(P_{n, B_n}^0))^2 - (X^{(j)}X^{(l)} -
        \psi_0^{(jl)})^2 \bigg\vert P_{n, B_n}^0, B_n\right] \\
      & = \sum_{j=1}^J\sum_{l=1}^J
        \mathbb{E}_{P_0}\bigg[(\psi_0^{(jl)} -
        \hat{\Psi}^{(jl)}_k(P_{n, B_n}^0))(2X^{(j)}X^{(l)} -
        \psi_0^{(jl)}\\
        & \qquad\qquad\qquad\qquad\qquad - \hat{\Psi}^{(jl)}_k(P_{n,
        B_n}^0)) \bigg\vert P_{n, B_n}^0, B_n\bigg] \\
      & = \sum_{j=1}^J\sum_{l=1}^J (\psi_0^{(jl)} -
        \hat{\Psi}^{(jl)}_k(P_{n, B_n}^0))\\
        & \qquad\qquad\qquad\qquad\qquad \mathbb{E}_{P_0}\bigg[2X^{(j)}X^{(l)} -
        \psi_0^{(jl)} - \hat{\Psi}^{(jl)}_k(P_{n,
        B_n}^0) \bigg\vert P_{n, B_n}^0, B_n\bigg] \\
           & =  \sum_{j=1}^J\sum_{l=1}^J (\psi_0^{(jl)}-
           \hat{\Psi}^{(jl)}_k(P_{n, B_n}^0))^2\\
      & =  \sum_{j=1}^J\sum_{l=1}^J \mathbb{E}_{P_0}\left[
        (\psi_0^{(jl)}-\hat{\Psi}^{(jl)}_k(P_{n, B_n}^0))^2
        \Big\vert P_{n, B_n}^0, B_n \right], \\
    \end{split}
  \end{equation*}
  where the second to last equality follows by noting that
  $\mathbb{E}_{P_0}[X^{(j)}X^{(l)}|P_{n, B_n}^0, B_n] =
  \mathbb{E}_{P_0}[X^{(j)}X^{(l)}]$ and that, by definition,
  $\mathbb{E}_{P_0}[X^{(j)}X^{(l)}] = \psi_0^{(jl)}$. Then,

  \begin{equation*}
    \begin{split}
      \text{Var}_{P_0}\left[Z_k \mid P^0_{n, B_n}, B_n\right] & \leq
        \mathbb{E}_{P_0}\left[Z_k^2 \mid P^0_{n, B_n}, B_n\right] \\
      & = \mathbb{E}_{P_0}\Bigg[\bigg( \sum_{j=1}^J\sum_{l=1}^J
        (\psi_0^{(jl)} - \hat{\Psi}^{(jl)}_k(P_{n, B_n}^0))
        (2X^{(j)}X^{(l)} - \psi_0^{(jl)} \\
      & \qquad \qquad \qquad \qquad \qquad  \qquad
        - \hat{\Psi}^{(jl)}_k(P_{n, B_n}^0))\bigg)^2 \Bigg\vert
        P^0_{n, B_n}, B_n\Bigg] \\
    & = J^4 \mathbb{E}_{P_0}\Bigg[\bigg(
        \frac{1}{J^2} \sum_{j=1}^J\sum_{l=1}^J
        (\psi_0^{(jl)} - \hat{\Psi}^{(jl)}_k(P_{n, B_n}^0))
        (2X^{(j)}X^{(l)} - \psi_0^{(jl)} \\
      & \qquad \qquad \qquad \qquad \qquad  \qquad
        - \hat{\Psi}^{(jl)}_k(P_{n, B_n}^0))\bigg)^2 \Bigg\vert
        P^0_{n, B_n}, B_n\Bigg] \\
    & \leq J^4 \mathbb{E}_{P_0}\Bigg[
        \frac{1}{J^2} \sum_{j=1}^J\sum_{l=1}^J
        (\psi_0^{(jl)} - \hat{\Psi}^{(jl)}_k(P_{n, B_n}^0))^2
        (2X^{(j)}X^{(l)} - \psi_0^{(jl)} \\
      & \qquad \qquad \qquad \qquad \qquad  \qquad
        - \hat{\Psi}^{(jl)}_k(P_{n, B_n}^0))^2
        \Bigg\vert P^0_{n, B_n}, B_n\Bigg] \\
    & \leq J^2 4(M_1 + M_2)^2 \mathbb{E}_{P_0}\left[
        \sum_{j=1}^J\sum_{l=1}^J (\psi_0^{(jl)} -
        \hat{\Psi}^{(jl)}_k(P_{n, B_n}^0))^2
        \Bigg\vert P^0_{n, B_n}, B_n\right] \\
    & = 4(M_1+M_2)^2J^2 \mathbb{E}_{P_0}\left[
        Z_k \mid P^0_{n, B_n}, B_n\right] \\
    & = \overline{M}(J) \mathbb{E}_{P_0}\left[
        Z_k \mid P^0_{n, B_n}, B_n\right]. \\
    \end{split}
  \end{equation*}
  Here, the second inequality holds from the application of Jensen's Inequality
  to the square of the double sum, which is effectively an expectation of a
  discrete uniform random variable when scaled by $J^2$. The final inequality
  results from Assumptions~1 and 2, and concludes the proof.
\end{proof}

The following lemma is a result taken directly from
\citet{dudoit2005asymptotics}. It is restated here for convenience.
\begin{lemma}\label{lma2:convergence-in-prob}

  Let $Y_1, Y_2, \ldots$ be a sequence of random variables. If
  $\mathbb{E}[\lvert Y_n \rvert] = O(g(n))$ for some positive function
  $g(\cdot)$, then $Y_n = O_P(g(n))$.

\end{lemma}
\begin{proof}

  We must show that, for each $\epsilon > 0$, there exists an $N$ and $B > 0$
  such that $\mathbb{P}(\lvert Y_n \rvert / g(n) > B) < \epsilon$ for all $n >
  N$. By assumption, there exists an $N$ and a $C > 0$ such that
  $\mathbb{E}[\lvert Y_n \rvert] / g(n) < C$ for all $n > N$. By defining $C /
  B = \epsilon$ and making use of Markov's Inequality, we find that
  \begin{equation*}
    \mathbb{P}\left(\frac{\lvert Y_n \rvert}{g(n)} > B \right) \leq 
    \frac{\mathbb{E}\left[\lvert Y_n \rvert\right]}{g(n)B} \leq \frac{C}{B} =
    \epsilon.
  \end{equation*}
\end{proof}

Having derived Lemma~\ref{lma1:bounded-loss-diff}, the remainder of the proof
for Theorem~\ref{paper:thm1} closely follows the proof of the first theorem in
\citet{dudoit2005asymptotics}.

\begin{proof} 
  \textbf{Theorem~\ref{paper:thm1}, Finite-Sample Result.}

\begin{equation}\label{split-risk-diff}
    \begin{split}
      0 & \leq \tilde{\theta}_{p_n,n}(\hat{k}_{p_n,n}) - \theta_0 \\
      & = \mathbb{E}_{B_n}\bigg[
        \int (L(x; \hat{\Psi}_{\hat{k}_{p_n,n}}(P_{n, B_n}^0)) -
        L(x; \boldsymbol{\psi}_0)) dP_0(x) \\
      & \qquad\qquad - (1 + \delta) \int L((x;
        \hat{\Psi}_{\hat{k}_{p_n,n}}(P_{n, B_n}^0)) -
        L(x; \boldsymbol{\psi}_0)) dP^1_{n, B_n}(x) \\
      & \qquad\qquad + (1 + \delta) \int
        (L(x; \hat{\Psi}_{\hat{k}_{p_n,n}}(P_{n, B_n}^0)) -
        L(x; \boldsymbol{\psi}_0)) dP^1_{n, B_n}(x) \bigg] \\
      & \leq \mathbb{E}_{B_n}\bigg[
        \int (L(x; \hat{\Psi}_{\hat{k}_{p_n,n}}(P_{n, B_n}^0)) -
        L(x; \boldsymbol{\psi}_0)) dP_0(x) \\
      & \qquad\qquad - (1 + \delta) \int (L(x;
        \hat{\Psi}_{\hat{k}_{p_n,n}}(P_{n, B_n}^0)) -
        L(x; \boldsymbol{\psi}_0)) dP^1_{n, B_n}(x) \\
      & \qquad\qquad + (1 + \delta) \int
        (L(x; \hat{\Psi}_{\tilde{k}_{p_n,n}}(P_{n, B_n}^0)) -
        L(x; \boldsymbol{\psi}_0)) dP^1_{n, B_n}(x) \bigg] \\
      & = \mathbb{E}_{B_n}\bigg[
        \int (L(x; \hat{\Psi}_{\hat{k}_{p_n,n}}(P_{n, B_n}^0)) -
        L(x; \psi_0)) dP_0(x) \\
      & \qquad\qquad - (1 + \delta) \int (L(x;
        \hat{\Psi}_{\hat{k}_{p_n,n}}(P_{n, B_n}^0)) -
        L(x; \boldsymbol{\psi}_0)) dP^1_{n, B_n}(x) \\
      & \qquad\qquad + (1 + \delta) \int
        (L(x; \hat{\Psi}_{\tilde{k}_{p_n, n}}(P_{n, B_n}^0)) -
        L(x; \boldsymbol{\psi}_0)) dP^1_{n, B_n}(x) \\
      & \qquad\qquad - (1 + 2\delta) \int
        (L(x; \hat{\Psi}_{\tilde{k}_{p_n,n}}(P_{n, B_n}^0)) -
        L(x; \boldsymbol{\psi}_0)) dP_0(x) \\
      & \qquad\qquad + (1 + 2\delta) \int
        (L(x; \hat{\Psi}_{\tilde{k}_{p_n,n}}(P_{n, B_n}^0)) -
        L(x; \boldsymbol{\psi}_0)) dP_0(x) \bigg]. \\
    \end{split}
  \end{equation}

  The first inequality is by assumption, and the second is by definition of
  $\hat{k}_{p_n,n}$ such that $\hat{\theta}_{p_n,n}(\hat{k}_{p_n,n}) \leq
  \hat{\theta}_{p_n,n}(k)$ $\forall k$. For simplicity in the remainder of the proof, we replace
  $\hat{k}_{p_n,n}$ and $\tilde{k}_{p_n,n}$ with $\hat{k}$ and $\tilde{k}$,
  respectively, in a slight abuse of notation.
  
  Now, let the first two terms of the last expression in
  Equation~\eqref{split-risk-diff} be denoted by $R_{\hat{k},n}$, and the
  third and fourth terms by $T_{\tilde{k},n}$. The last term is the
  cross-validated oracle's risk difference: $(1 +
  2\delta)(\tilde{\theta}_{p_n,n}(\tilde{k}) - \theta_0)$.
  Thus,
  
  \begin{equation}\label{target-finite-res}
    0 \leq \tilde{\theta}_{p_n,n}(\hat{k}) - \theta_0 \leq
    (1 + 2\delta)(\tilde{\theta}_{p_n,n}(\tilde{k})
    - \theta_0) + R_{\hat{k},n} + T_{\tilde{k},n}.
  \end{equation}

  We next show that $\mathbb{E}_{P_0}[R_{\hat{k},n}+
  T_{\tilde{k},n}]~\leq~2c(\delta,\overline{M}(J))(1+\text{log}(K))/(np_n)$,
  where  $c(\delta,\overline{M}(J))=2(1+\delta)^2\overline{M}(J)(1/\delta+1/3)$
  for some $\delta > 0$. For convenience, let 
  
  \begin{equation*}
    \begin{split}
      \hat{H}_k & \coloneqq \int (L(x;
        \hat{\Psi}_{k}(P_{n, B_n}^0)) -
        L(x; \boldsymbol{\psi}_0)) dP^1_{n, B_n}(x) \\
      \tilde{H}_k & \coloneqq \int
        (L(x; \hat{\Psi}_{k}(P_{n, B_n}^0)) -
        L(x; \boldsymbol{\psi}_0)) dP_0(x) \\
      R_{k,n}(B_n) & \coloneqq (1 + \delta)(\tilde{H}_k -
        \hat{H}_k) - \delta\tilde{H}_k \\
        T_{k,n}(B_n) & \coloneqq (1 + \delta)(\hat{H}_k -
        \tilde{H}_k) - \delta\tilde{H}_k,
    \end{split}
  \end{equation*}
  so that $R_{k,n} = \mathbb{E}_{B_n}[R_{k,n}(B_n)]$ and $T_{k,n} = 
  \mathbb{E}_{B_n}[T_{k,n}(B_n)]$.

  Given $B_n$ and $P_{n, B_n}^0$, let $Z_{k, i}, \; 1 \leq i \leq np_n$, denote
  the $np_n$ i.i.d. copies of $Z_k$ corresponding with the validation set,
  that is, with $\{X_i : B_n(i) = 1\}$ (as defined in
  Lemma~\ref{lma1:bounded-loss-diff}). Then, $\hat{H}_k = \sum_i Z_{k,i}/np_n$
  and $\tilde{H}_k = \mathbb{E}_{P_0}[Z_{k, i} \mid P^{0}_{n, B_n}, B_n]$.
  Hence, $\tilde{H}_k - \hat{H}_k$ is an empirical mean of $np_n$
  i.i.d.~centered random variables. Further, by Assumptions~1 and
  2, $\lvert Z_{k,i} \rvert < 2(M_1 + M_2)^2J^2$ a.s..
  Next, we apply Bernstein's Inequality to the centered empirical mean
  $\tilde{H}_k - \hat{H}_k$, using the property of the $Z_{k,i}$'s derived in
  Lemma~\ref{lma1:bounded-loss-diff}, to obtain a bound for the tail
  probabilities of $R_{k, n}(B_n)$ and $T_{k, n}(B_n)$: 

  \begin{equation*}
    \sigma^2_k \coloneqq \text{Var}_{P_0}\left[Z_k \mid P^0_{n, B_n}, B_n\right]
      \leq \overline{M}(J) \mathbb{E}\left[Z_k \mid P^0_{n, B_n}, B_n\right]
      = \overline{M}(J) \tilde{H}_{k}.
  \end{equation*}
  Then, for $s > 0$, Bernstein's Inequality yields
  \begin{equation*}
    \begin{split}
      \mathbb{P}_{P_0}\left(R_{k,n}(B_n) > s \mid P_{n, B_n}^0, B_n\right)
        & = \mathbb{P}_{P_0}\left(\tilde{H}_k - \hat{H}_k >
          \frac{s + \delta\tilde{H}_k}{1+ \delta} \Bigg\vert
          P_{n, B_n}^0, B_n\right) \\
        & \leq \mathbb{P}_{P_0}\left(\tilde{H}_k - \hat{H}_k >
          \frac{s + \delta\sigma_k^2/\overline{M}(J)}{1+ \delta}
          \Bigg\vert P_{n, B_n}^0, B_n\right) \\
        & \leq \text{exp}\left\{
          -\frac{np_n}{2(1+\delta)^2}\frac{\left(s + \delta
          \sigma_k^2/\overline{M}(J)\right)^2}
          {\sigma_k^2 + \frac{\overline{M}(J)}{3(1 + \delta)}\left(s + \delta 
          \sigma_k^2/\overline{M}(J)\right)}\right\}.
    \end{split}
  \end{equation*}
  Note that
  \begin{equation*}
    \frac{\left(s + \delta \sigma_k^2/\overline{M}(J)\right)^2}
    {\sigma_k^2 + \frac{\overline{M}(J)}{3(1 + \delta)}\left(s + \delta
      \sigma_k^2/ \overline{M}(J)\right)}
    = \frac{s + \delta \sigma_k^2/\overline{M}(J)}{
      \frac{\sigma_k^2}{s + \delta \sigma_k^2/\overline{M}(J)} + 
      \frac{\overline{M}(J)}{3(1 + \delta)}
    }
    \geq \frac{s + \delta \sigma_k^2/\overline{M}(J)}{
      \overline{M}(J)\left(\frac{1}{\delta} + \frac{1}{3}\right)
    }
    \geq  \frac{s}{\overline{M}(J)\left(\frac{1}{\delta} + \frac{1}{3}\right)}.
  \end{equation*}
  And so, for $s > 0$,
  \begin{equation*}
    \mathbb{P}_{P_0}\left(R_{k,n}(B_n) > s \mid P_{n, B_n}^0, B_n\right) \leq 
      \text{exp}\left\{ -\frac{np_n}{c\left(\delta, \overline{M}(J)\right)}s
      \right\} \leq K \; \text{exp}\left\{-\frac{np_n}{c\left(\delta,
      \overline{M}(J)\right)}s\right\},
  \end{equation*}
  where $c(\delta, \overline{M}(J)) = 2(1 + \delta)^2
  \overline{M}(J)(\frac{1}{\delta} + \frac{1}{3})$. The same bound applies for
  the marginal probabilities of $\mathbb{P}_{P_0}(R_{k,n}(B_n) > s)$ since they
  hold for arbitrary choices of $B_n$ and $P_{n,B_n}^{0}$. The second
  inequality follows from $K$ being larger than or equal to $1$ by
  definition.

  Finally, for any $u > 0$, we have by the properties of expectations and the
  previously derived result that
  \begin{equation*}
    \begin{split}
      \mathbb{E}_{P_0}[R_{\hat{k},n}]
        & = \int_0^\infty \mathbb{P}_{P_0}(R_{\hat{k},n} > s)ds
        - \int_{-\infty}^0 \mathbb{P}_{P_0}(R_{\hat{k},n} \leq
        s)ds \\
      & \leq \int_0^\infty \mathbb{P}_{P_0}
        (R_{\hat{k},n} > s)ds \\
      & \leq u + \int_u^\infty \mathbb{P}_{P_0} (R_{\hat{k},n} > s)ds \\
      & \leq u + \int_u^{\infty}
        K \text{exp}\left\{ - \frac{np_n}{c\left(\delta,
          \overline{M}(J)\right)} s\right\}ds.
    \end{split}
  \end{equation*}
  
  Since the expression on the right-hand side of the inequality above
  achieves its minimum value of $c(\delta, \overline{M}(J))(1
  + \text{log}(K)/(np_n)$ at $u_n = c(\delta,
  \overline{M}(J))\text{log}(K)/(np_n)$, then

  \begin{equation*}
    \mathbb{E}_{P_0}[R_{\hat{k}, n}] \leq c\left(\delta,
    \overline{M}(J)\right) \frac{1 + \text{log}(K)}{np_n}.
  \end{equation*}
  The same bound applies to $\mathbb{E}_{P_0}[T_{\tilde{k}, n}]$.
  Therefore, taking the expected values of the inequality in
  Equation~\eqref{target-finite-res}, we produce the desired finite-sample
  result in Equation~\eqref{paper:thm1:finite-sample-result} of the main text.
\end{proof}

\begin{proof}
  \textbf{Theorem~\ref{paper:thm1}, High-Dimensional Asymptotic Result.}

  The expected risk differences ratio's convergence follows directly from the
  main text's Equation~\eqref{paper:thm1:finite-sample-result} for some $\delta
  > 0$, so long as $c(\delta, \overline{M}(J))(1 + \text{log}(K))/(np_n
  \mathbb{E}_{P_0}[\tilde{\theta}_{p_n,n}(\tilde{k}_{p_n, n}) - \theta_0])
  \rightarrow 0$ as $n, J \rightarrow \infty$. Given the assumption in
  Kolmogorov asymptotics that $J/n \rightarrow m < \infty$ as $J, n \rightarrow
  \infty$, an equivalent condition is that $m(M_1+M_2)^2J(1 +
  \text{log}(K))/(p_n \mathbb{E}_{P_0}[\tilde{\theta}_{p_n,n}(\tilde{k}_{p_n,
  n}) - \theta_0]) \rightarrow 0$ as $n, J \rightarrow \infty$. Convergence in
  probability then follows from Lemma~\ref{lma2:convergence-in-prob}.

\end{proof}

Though there are minor adaptations to the assumptions to reflect the use of
high-dimensional asymptotics, the proof of
Corollary~\ref{paper:full-dataset-result} follows that of
Corollary~\ref{paper:full-dataset-result} in \citet{dudoit2005asymptotics}.

\begin{proof} \textbf{Corollary~\ref{paper:full-dataset-result}.}

  The asymptotic statement of the main text's
  Equation~\eqref{paper:cor1-asymptotic} is an immediate result of
  Theorem~\ref{paper:thm1}'s Equation~\eqref{paper:thm1:convergence-in-prob}.

  \begin{equation*}
    \frac{\tilde{\theta}_{p_n,n}(\hat{k}_{p_n,n}) - \theta_0}
    {\tilde{\theta}_{p_n,n}(\tilde{k}_{p_n,n})- \theta_0}
    \frac{\tilde{\theta}_{p_n,n}(\tilde{k}_{p_n,n})-\theta_0}
    {\tilde{\theta}_{n}(\tilde{k}_n)-\theta_0}
    \overset{P}{\rightarrow} 1.
  \end{equation*}

  Letting $Z_{1,n} \coloneqq (n(1-p_n))^\gamma(\tilde{\theta}_{p_n,n}(
  \tilde{k}_{p_n,n})-\theta_0)$ and $Z_{2, n} \coloneqq
  n^\gamma(\tilde{\theta}_{n}(\tilde{k}_n)-\theta_0)$, and assuming that the
  sufficient condition in Equation~\eqref{paper:suff-cond} of the main text
  holds, we find that $Z_{1, n} / Z_{2, n} \overset{d}{\rightarrow} 1$ by the
  Continuous Mapping Theorem. Then, notice that
  \begin{equation*}
    \frac{Z_{1, n}}{Z_{2, n}} =
    \frac{(1-p_n)^\gamma(\tilde{\theta}_{p_n,n}(
    \tilde{k}_{p_n,n})-\theta_0)}
    {\tilde{\theta}_{n}(\tilde{k}_n)-\theta_0},
  \end{equation*}
  which yields the desired sufficient condition when $p_n \rightarrow 0$. In
  the case of single-split validation, $Z_{2,n} \overset{d}{=} Z_{1,
  n/(1-p_n)}$, and so $Z_{2, n} \overset{d}{\rightarrow} Z$ implies that
  $(Z_{1, n}, Z_{2, n}) \overset{d}{\rightarrow} (Z, Z)$.

\end{proof}

\rone{
\begin{remark}
  \label{remark:unknown-mean}
  We assume throughout this work that $\mathbb{E}_{P_0}[X] = 0$ without loss of
  generality. In practice, however, the mean vector is generally unknown.
  Consider the uniformly bounded random vector $Y$ such that $X = Y -
  \mathbb{E}_{P_0}[Y]$. We might therefore consider using the demeaned random
  vector $\tilde{Y} = Y - \bar{Y}$ instead, where $\bar{Y}^{(j)} = 1/n \sum
  Y_i^{(j)}$. Employing $\tilde{Y}$ in place of $X$ in
  Lemma~\ref{lma1:bounded-loss-diff}, and denoting $\tilde{Z}_k \coloneqq
  L(\tilde{Y}; \hat{\Psi}_k(P_{n,B_n}^0)) - L(\tilde{Y}; \boldsymbol{\psi}_0)$,
  we find that $\mathbb{E}_{P_0}[\tilde{Z}_k | P_{n,B_n}^0, B_n] =
  \sum\sum(\psi_0^{(jl)}-\hat{\Psi}_k^{(jl)}(P_{n,B_n}^0))
  ((n-2)\psi_0^{(jl)}/n-\hat{\Psi}_k^{(jl)}(P_{n,B_n}^0))$. It then follows
  that, as $n \rightarrow \infty$, $\mathbb{E}_{P_0}[\tilde{Z}_k | P_{n,B_n}^0,
  B_n] = \mathbb{E}_{P_0}[Z_k | P_{n,B_n}^0, B_n]$. The asymptotic results of
  Theorem~\ref{paper:thm1} and Corollary~\ref{paper:full-dataset-result} are
  therefore achievable when $\mathbb{E}_{P_0}[Y]$ is unknown. The same cannot
  be said for the finite sample result of Theorem~\ref{paper:thm1}:
  $\mathbb{E}_{P_0}[\tilde{Z}_k | P_{n,B_n}^0, B_n]$ is not strictly
  nonegative. For large enough values of $n$, however, we do expect these
  finite bounds to be approximately correct.
\end{remark}
}

\section{Additional Figures and Tables}

\FloatBarrier

\begin{table}
  \centering
  \begin{tabular}{c c} 
    \hline
    Estimator & Hyperparameters \\ [0.5ex] 
    \hline\hline
    Sample covariance matrix & Not applicable \\ 
    \hline
    Hard thresholding \citep{bickelb:2008} & ${\text{Thresholds}=\{0.1, 0.2, \ldots, 1.0\}}$ \\
    \hline
    SCAD thresholding \citep{fan:2001,rothman:2009} & ${\text{Thresholds} = \{0.1, 0.2, \ldots, 1.0\}}$ \\
    \hline
    Adaptive LASSO \citep{rothman:2009} & ${\text{Thresholds} =
    \{0.1,0.2,\ldots,0.5\}}$; \\
                                        & ${\text{exponential weights} = \{0.1,0.2,\ldots,0.5\}}$  \\
    \hline
    Banding \citep{bickela:2008} & $\text{Bands} = \{1, 2, \ldots, 5\}$  \\
    \hline
    Tapering \citep{cai:2010} & $\text{Bands} = \{2, 4, \ldots, 10\}$\\
    \hline
    Linear shrinkage \citep{ledoit:2004} & Not applicable \\
    \hline
    Dense linear shrinkage \citep{shafer:2005} & Not applicable \\
    \hline
    Nonlinear shrinkage \citep{ledoit:2020b} & Not applicable \\
    \hline
    POET \citep{fan:2013} using hard thresholding & ${\text{Latent
    factors}~=~\{1,2,\ldots,5\}}$; \\
                                                  & ${\text{thresholds}~=~\{0.1,0.2,0.3\}}$ \\
    \hline
  \end{tabular}
  \caption{Families of candidate estimators used by cvCovEst in the simulation
    study ($74$ distinct estimators in total)}
  \label{table:sim-hyperparams}
\end{table}

\begin{figure}
  \centering
  \includegraphics[width=0.7\textwidth]{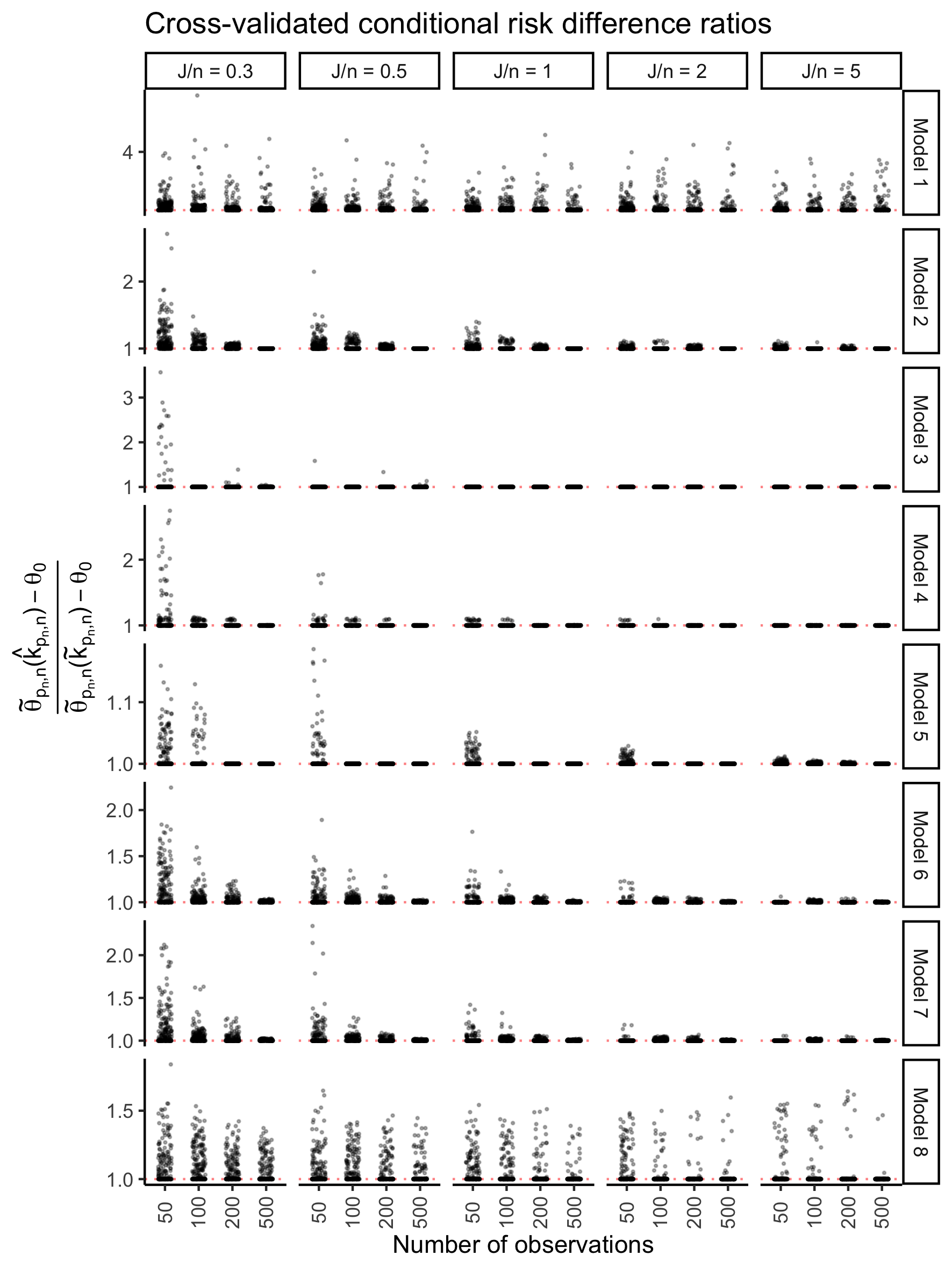}
  \caption{
  Comparison of the cross-validated ($\hat{k}_{p_n,n}$) and cross-validated
  oracle ($\tilde{k}_{p_n,n}$) selections' cross-validated conditional risk
  differences. The proposed cross-validated selection procedure achieves
  asymptotic equivalence in most settings for relatively small sample sizes and
  numbers of features.}
  \label{sim:conv-in-prob}
\end{figure}

\begin{figure}
  \centering
  \includegraphics[width=0.8\textwidth]{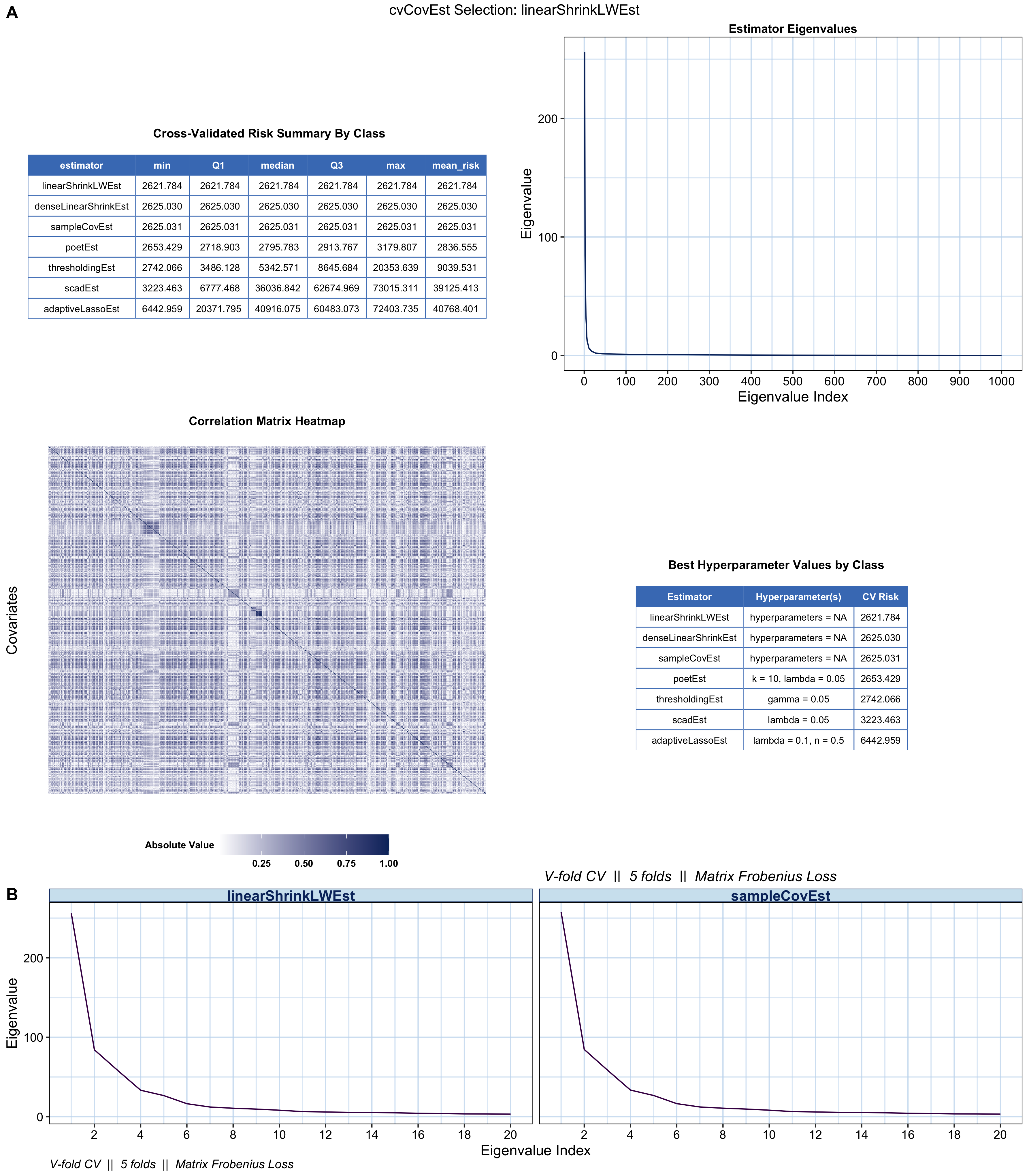}
  \caption{Zeisel dataset: Diagnostic plots and tables generated using the
  cvCovEst R package.
  The components in \textbf{(A)} can be interpreted in the same manner as the
  previous figure, with the exception of the top-left panel. In this table, a
  five number summary of the cross-validated Frobenius risk is given for each
  class of estimators considered across all possible combinations of
  hyperparameters, if any. It is clear from the tables in \textbf{(A)} that the
  cvCovEst selection is essentially equivalent in terms of cross-validated risk
  to the sample covariance matrix. The plots in \textbf{(B)} further highlight
  that the 20 leading eigenvalue of the cvCovEst estimate and the sample
  covariance matrix are indistinguishable.}
  \label{fig:zeisel-diagnostics}
\end{figure}

\begin{table}
  \centering
  \begin{tabular}{c c c} 
    \hline
    Estimator & Hyperparameters \\ [0.5ex] 
    \hline\hline
    Sample covariance matrix & NA \\ 
    \hline
    Hard thresholding & Thresholds $= \{0.05, 0.10, \ldots, 0.30\}$ \\
    \hline
    SCAD thresholding & Thresholds $ = \{0.05, 0.10, \ldots, 0.50\}$ \\
    \hline
    Adaptive LASSO & Thresholds $= \{0.1, 0.2, \ldots, 0.5\}$; \\
                   & exponential weights $ = \{0.1, 0.2, \ldots, 0.5\}$  \\
    \hline
    Linear shrinkage & NA \\
    \hline
    Dense linear shrinkage & NA \\
    \hline
    POET using hard thresholding & Latent factors$ = \{5, 6, \ldots, 10\}$; \\
                                 & thresholds $= \{0.05, 0.10, \ldots, 0.3\}$ \\
    \hline
  \end{tabular}
  \caption{Families of candidate estimators used in single-cell
  transcriptomic data analyses}
  \label{table:scrnaseq-hyperparams}
\end{table}

\FloatBarrier
\section{Extended Simulation Results}
\label{spectral-norm-sim}

We compared the estimates generated by our method against those of the
individual candidate procedures using the simulated data sets. This was
accomplished by computing the Frobenius norm of each estimate against the
corresponding true covariance matrix. The mean norms over all simulations were
then computed for each covariance matrix estimation procedure, again stratified
by $n$, $J/n$, and the covariance matrix model (Figure~\ref{mean-frobenius}).
Our CV scheme was used to select hyperparameters of these competing approaches
where necessary. As stated previously, our CV approach is generally equivalent
to these estimators' hyperparameter selection procedures. The hyperparameters
considered are provided in Table~\ref{table:norm-hyperparams}. Where
appropriate, the competing methods' hyperparameters are more varied than those
used by cvCovEst, reflecting more aggressive estimation procedures that one
might employ when only using a single family of estimators.

\begin{table}
  \centering
  \begin{tabular}{c c c} 
    \hline
    Estimator & Hyperparameters \\ [0.5ex] 
    \hline\hline
    Sample covariance matrix & NA \\ 
    \hline
    Hard thresholding & Thresholds $= \{0.05, 0.10, \ldots, 1.00\}$ \\
    \hline
    SCAD thresholding & Thresholds $= \{0.05, 0.10, \ldots, 1.00\}$ \\
    \hline
    Adaptive LASSO & Thresholds $= \{0.1, 0.2, \ldots, 0.5\}$; \\
                   & exponential weights $= \{0.1, 0.2, \ldots, 0.5\}$  \\
    \hline
    Banding & Bands $= \{1, 2, \ldots, 10\}$  \\
    \hline
    Tapering &  Bands $ = \{2, 4, \ldots, 10\}$\\
    \hline
    Linear shrinkage & NA \\
    \hline
    Dense linear shrinkage & NA \\
    \hline
    Nonlinear shrinkage & NA \\
    \hline
    POET using hard thresholding & Latent factors $ = \{1, 2, \ldots, 10\}$; \\
                                 & thresholds $= \{0.1, 0.2, \ldots, 1.0\}$ \\
    \hline
  \end{tabular}
  \caption{Families of candidate estimators compared against the cross-validated
  loss-based estimator selection procedure. Note that the library of candidate
  estimators used by the proposed method is provided in
  Table~\ref{table:sim-hyperparams}}
  \label{table:norm-hyperparams}
\end{table}

\begin{figure}
  \centering
  \includegraphics[width=\textwidth]{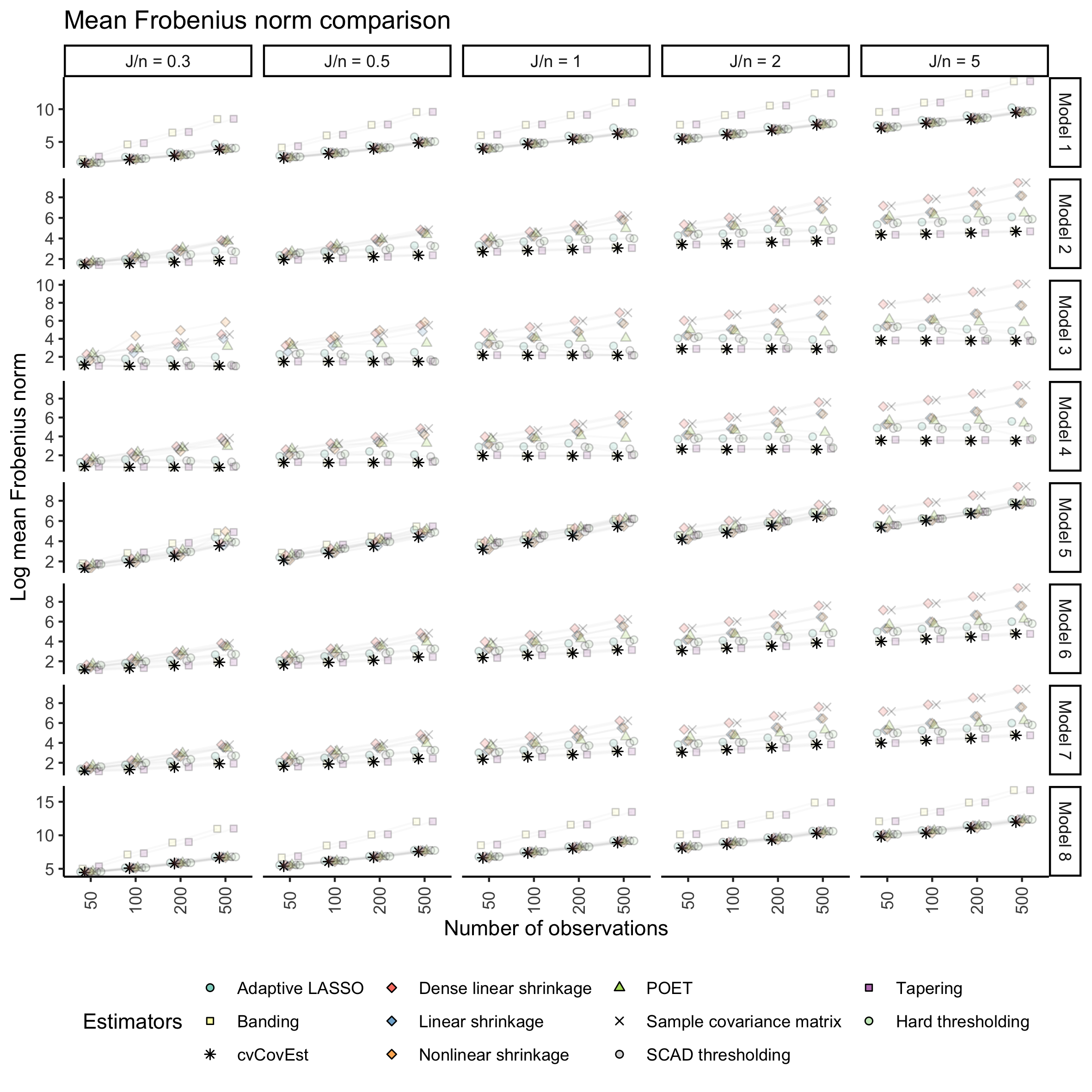}
  \caption{
    Comparison of competing, bespoke covariance matrix estimation procedures to
    our cross-validated selection approach in terms of the Monte Carlo mean
    Frobenius norm under a variety of data-generating processes. Note that the
    scales of the y-axis are tailored to the covariance matrix model.
  }
  \label{mean-frobenius}
\end{figure}

We repeated this benchmarking experiment using the spectral norm to assess the
accuracy of our estimation procedure with respect to the leading eigenvalue of
the covariance matrix. Recall that the spectral norm of a square matrix is
defined as it's largest absolute eigenvalue. Though our theoretical results do
not relate to to this norm, outcomes similar to those in
Figure~\ref{mean-frobenius} are expected given the relationship between these
two norms for reasons previously described.

\begin{figure}
  \centering
  \includegraphics[width=\textwidth]{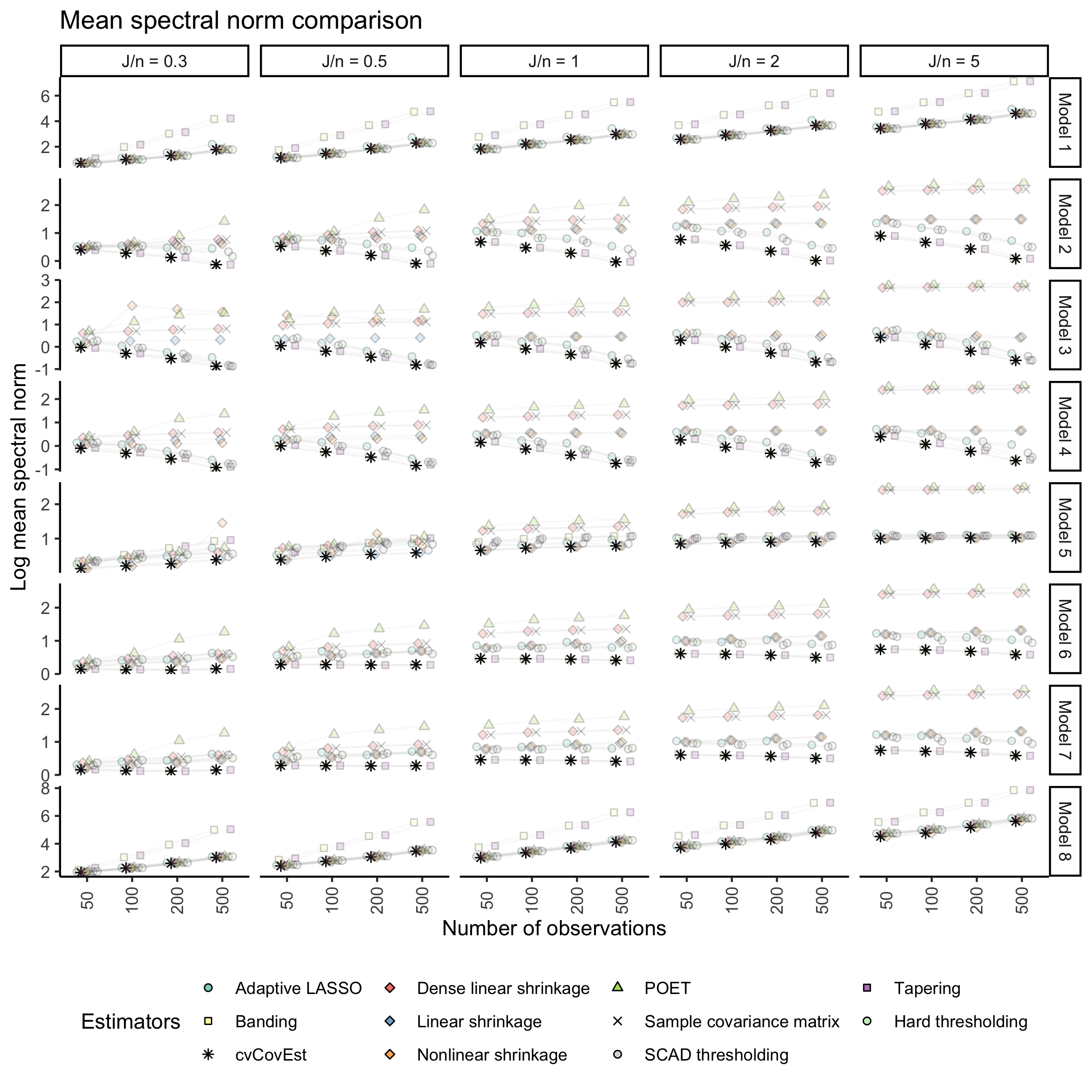}
  \caption{Comparison of competing, bespoke covariance matrix estimation
    procedures to our cross-validated selection approach in terms of the Monte
    Carlo mean spectral norm under a variety of data-generating processes. Note
    that the scales of the y-axis are tailored to the covariance matrix model.}
  \label{mean-spectral}
\end{figure}

The results, presented in Figures~\ref{mean-frobenius} and \ref{mean-spectral},
demonstrate that our estimator selection procedure performs at least as well as
the best alternative estimation strategy. This suggests that procedures
dedicated to or relying upon the accurate estimation of leading eigenvalues and
eigenvectors, like principal component analysis and latent variable estimation,
might benefit from the integration of our cross-validated covariance matrix
estimation framework.

\FloatBarrier

\section{Candidate Covariance Matrix Estimators}\label{lit}

Using the proposed cross-validated selection procedure effectively requires
a large and diverse set of candidate covariance matrix estimators. In
this spirit, we provide in the sequel a brief overview of select covariance
matrix estimators that have proven to be effective in a variety of settings.
Note, however, that the proposed selection framework need not be limited to
those described here. Thorough reviews of estimators have been provided
by~\citet{pourahmadi:2013}, ~\citet{fan:2016},
and~\citet{ledoit:2020a}, to name only a few.

\subsection{Thresholding Estimators}\label{thresh-est}

An often natural, simplifying assumption about the true covariance matrix's
structure is that it is sparse, that is, a non-negligible portion of its
off-diagonal elements have a value near zero or equal to zero. This assumption
is not altogether unreasonable: Given a system of numerous variables, it seems
unlikely in many settings that a majority of these variables would depend on one
another.

The class of generalized thresholding estimators~\citep{bickelb:2008,
rothman:2009,cai:2011} is one collection of covariance matrix estimators
based upon this structural assumption. Given the sample covariance matrix
$\mathbf{S}_n$, the generalized thresholding estimator is defined as
\begin{equation*}
  \hat{\Psi}_{\text{gt}}(P_n; t) \coloneqq \left\{t\left(S_n^{(jl)}\right):
  \; j,l = 1, \ldots, J\right\}, 
\end{equation*}
where $t(\cdot)$ is a thresholding function that often requires one or more
hyperparameters dictating the amount of regularization. The hard, soft,
smoothly clipped absolute deviation (SCAD) \citep{fan:2001}, and adaptive
LASSO \citep{rothman:2009} functions are among the collection of suitable
thresholding operators. As an example, the hard-thresholding function is
defined as $t_u(S_n^{(jl)}) \equiv S_n^{(jl)} I(S_n^{(jl)} > u)$ for some
threshold $u > 0$. \citet{cai:2011} also demonstrated that element-specific
thresholding functions might be useful when the features' variances are highly
variable. The hyperparameters for any specific thresholding function are often
selected via CV. Regardless of the choice of $t(\cdot)$, these estimators
preserve the symmetry of the sample covariance matrix and are invariant under
permutations of the features' order.

\citet{bickelb:2008}~and~\citet{rothman:2009} have shown that these
estimators are consistent under the spectral norm (which is defined as a square
matrix's largest absolute eigenvalue), assuming that $\log J/n \rightarrow 0$,
that the observations' marginal distributions satisfy a tail condition, and
that the true covariance matrix is a member of the class of matrices satisfying
a particular notion of ``approximate sparsity.'' \citet{cai:2011} have
derived similar results for their entry-specific thresholding estimator over an
even broader parameter space of sparse covariance matrices.

\subsection{Banding and Tapering Estimators}

A family of estimators related to thresholding estimators are banding and
tapering estimators~\citep{bickela:2008,cai:2010}. Like thresholding
estimators, these estimators rely on the assumption that the true covariance
matrix is sparse. However, the structural assumption of these estimators on
$\boldsymbol\psi_0$ is much more rigid than that of thresholding estimators.
Specifically, such estimators assume that the true covariance matrix is a
\textit{band matrix}, that is, a sparse matrix whose non-zero entries are
concentrated about the diagonal. These estimators therefore require a natural
ordering of the observations' features, operating under the assumption that
``distant'' variables are uncorrelated. Such structure is often encountered in
longitudinal and time-series data.

Given the sample covariance matrix  $\mathbf{S}_n$, the banding estimator of
\citet{bickela:2008} is defined as
\begin{equation*}
  \hat{\Psi}_{\text{band}}(P_n; b) \coloneqq \left\{
    S_n^{(jl)} \mathbb{I}(\lvert j - l \rvert \leq b):
    \; j,l = 1, \ldots, J \right\},
\end{equation*}
where $b$ is a hyperparameter that determines the number of bands to retain
from the sample covariance matrix and is chosen via a CV procedure. For the
class of ``bandable'' covariance matrices, i.e., the set of well-conditioned
matrices whose elements not in the central bands of the matrix decay rapidly,
this banding estimator has been shown to be uniformly consistent in the
spectral norm so long as $\log(J)/n \rightarrow 0$.

The tapering estimator of~\citet{cai:2010} is the smooth generalization of
the banding estimator, gradually shrinking the off-diagonal bands of the sample
covariance matrix towards zero. It is defined as
\begin{equation*}
    \hat{\Psi}_{\text{tap}}(P_n; b)
    \coloneqq \mathbf{W}_{b} \circ \mathbf{S}_n,\\
\end{equation*}
for some weight matrix $\mathbf{W}_b$. Here, ``$\circ$'' denotes the Hadamard
(element-wise) matrix product. Clearly, letting $W^{(jl)}_b = \mathbb{I}(\lvert
j - l \rvert \leq b)$ for some positive integer $b$ results in the banding
estimator. A popular weighting scheme derived by \citet{cai:2010} is
\begin{equation*}
  W^{(jl)}_b \coloneqq \begin{cases}
    1, &\text{when $\lvert j - l \rvert \leq \frac{b}{2}$}  \\
    2 - \frac{\lvert j - l \rvert}{b},
      &\text{when $\frac{b}{2} < \lvert j - l \rvert \leq b$} \\
    0, & \text{otherwise}
\end{cases}
,
\end{equation*}
which we use in our simulation study presented in Section~4 of the main text.
\citet{cai:2010} also derived the optimal rates of convergence for this
estimator under the Frobenius and spectral norms, considering a class of
bandable covariance matrices that is more general than that considered by
\citet{bickela:2008}: The smallest eigenvalue of the covariance matrices in
this class can take on a value of zero. However, this estimator does not
improve upon the bounds set by the banding estimator.

\subsection{Shrinkage Estimators}

We next consider the linear and non-linear shrinkage estimators inspired by
Stein's work on empirical Bayesian methods. These estimators are
rotation-equivariant, shrinking the eigenvalues of the sample covariance matrix
towards a set of target values, whilst leaving its eigenvectors untouched. In
doing so, the resultant estimators are better-conditioned than the
sample covariance matrix in a manner guaranteeing that the resultant covariance
matrix estimator be non-singular. Further, these estimators do not rely on 
sparsity assumptions about the true covariance matrix, setting them apart from
those previously discussed.

One of the first shrinkage estimators, the linear shrinkage estimator of the
sample covariance matrix, was proposed by~\citet{ledoit:2004}. This
estimator is defined as the convex combination of the sample covariance matrix
and the identity matrix. Hence, it represents a compromise between
$\mathbf{S}_n$, an unbiased but highly variable estimator of $\boldsymbol{\psi}_0$ in
high dimensions, and $\mathbf{I}_{J \times J}$, a woefully biased but fixed
estimator.  \citet{ledoit:2004} found that, under mild conditions, the
asymptotically optimal shrinkage intensity with respect to the scaled (by $J$)
Frobenius norm can be estimated consistently in high dimensions. This estimator
is defined as
\begin{equation*}
  \hat{\Psi}_{\text{ls}}(P_n) \coloneqq
    \frac{b^2_n}{d^2_n}m_n \mathbf{I} +
    \frac{a_n^2}{d_n^2} \mathbf{S}_n,
\end{equation*}
for $m_n = \text{tr}(\mathbf{S}_n) / J$,
$d^2_n = \lVert \mathbf{S}_n - m_n\mathbf{I} \rVert^2_{F, 1/J}$,
$\bar{b}^2_n = n^{-2} \sum_i \lVert X_i X_i^\top - \mathbf{S}_n
\rVert^2_{F, 1/J}$, $b^2_n = \min(\bar{b}^2_n, d^2_n)$, and $a^2_n = d^2_n -
b^2_n$.

Bespoke shrinkage targets may be used in place of the identity. For example,
one might consider a dense matrix target whose diagonal elements are the
average of the sample covariance matrix's diagonal elements and whose
off-diagonal elements are equal to the average of all the sample covariance
matrix's off-diagonal elements. For the sake of brevity, discussion of such
estimators is omitted, but examples are provided by, among
others,~\citet{ledoit:2003} and~\citet{shafer:2005}, particularly for
use in financial economics and statistical genomics applications, respectively.

When assumptions about the true covariance matrix's structure are unfounded, it
can become impossible to select an appropriate linear shrinkage target.
Instead, one might consider generalizing these estimators to shrink the
eigenvalues of the sample covariance matrix in a non-linear fashion. That is, an
estimator that shrinks the sample covariance matrix's eigenvalues not by a
common shrinkage factor (as with linear shrinkage estimators) but with
shrinkage factors tailored to each sample eigenvalue. As with the
aforementioned linear shrinkage estimators, such non-linear shrinkage
estimators produce positive-definite estimates so long as the shrunken sample
eigenvalues are positive and rotation-equivariant. These estimators belong to a
class initially introduced by~\citet{stein:1986} and have since seen a
resurgence in the work of~\citet{ledoit:2012,ledoit:2015}. More
recently,~\citet{ledoit:2020b} derived an analytical non-linear shrinkage
estimator that is asymptotically optimal in high dimensions and more
computationally efficient than their previously formulated estimators.

\subsection{Estimators Based on Factor Models}

Covariance matrix estimators based on factor models form another broad family
of estimators that do not assume sparsity of the true covariance matrix. Often
encountered in econometrics and finance, these estimators utilize the operating
assumption that the dataset's observations are functions of a few common, often
latent, factors. The factor model can be described as follows:
\begin{equation}\label{factor-model}
  X = \mu + \boldsymbol\beta F + \epsilon,
\end{equation}
where $X_{J \times 1}$ represents a random observation, $\mu_{J \times 1}$ a
mean vector, $\boldsymbol{\beta}_{J \times L}$ a matrix of factor coefficients,
$F_{L \times 1}$ a random vector of $L$ common factors, and $\epsilon_{J \times
1}$ a random error vector. Assuming that $F$ and $\epsilon$ are uncorrelated,
the covariance matrix of $X$ is given by
\begin{equation}\label{fm-cov}
  \boldsymbol{\psi} = \boldsymbol{\beta}\text{Cov}(F)\boldsymbol{\beta}^\top +
  \boldsymbol{\psi}_\epsilon,
\end{equation}
where $\psi_\epsilon$ is the covariance matrix of the random error.

For a review on estimating the covariance matrix in the presence of observable
factors, see \citet{fan:2016}. We now briefly discuss the estimation of
$\boldsymbol\psi$ when the factors are unobservable. Notice that when
$\boldsymbol{\psi}_\epsilon$ is not assumed to be diagonal, the decomposition
of $\boldsymbol{\psi}$ in Equation~\eqref{fm-cov} is not identifiable for fixed
$n$ and $J$, since the random vector $X$ is the only observed component of the
model. By letting $J \rightarrow \infty$, and assuming that the eigenvalues of
$\boldsymbol{\psi}_\epsilon$ are uniformly bounded or grow at a slow rate
relative to $J$ and that the eigenvalues of $(1/J)\;
\boldsymbol{\beta}^\top\boldsymbol{\beta}$ are uniformly bounded away from zero
and infinity, it can be shown that
$\boldsymbol{\beta}\text{Cov}(F)\boldsymbol{\beta}^\top$ is asymptotically
identifiable \citep{cai:2010}. It follows from these assumptions that the
signal in the factors increases as the number of features increases, while the
noise contributed by the error term remains constant. The eigenvalues
associated with $\boldsymbol{\beta}\text{Cov}(F)\boldsymbol{\beta}^\top$
therefore become easy to differentiate from those of
$\boldsymbol\psi_\epsilon$.

Now, even with $\boldsymbol{\beta}\text{Cov}(F)\boldsymbol{\beta}^\top$ being
asymptotically identifiable, $\boldsymbol{\beta}$ and $F$ cannot be
distinguished. As a solution, the following constraint is imposed on $F$:
$\text{Cov}(F) = \mathbf{I}_{L \times L}$. It then follows that
\begin{equation*}
  \boldsymbol{\psi} = \boldsymbol{\beta}\boldsymbol{\beta}^\top +
    \boldsymbol\psi_\epsilon.
\end{equation*}

Under the additional assumption that the columns of $\boldsymbol\beta$ be
orthogonal, \citet{fan:2013} found that the leading $L$ eigenvalues of
$\boldsymbol{\psi}$ are spiked, meaning that they are bounded below by some
constant \citep{johnstone2001distribution}, and grow at rate $O(J)$ as the
dimension of $\boldsymbol{\psi}$ increases. The remaining $J-L$ eigenvalues are
then either bounded above or grow slowly. This implies that the latent factors
and their loadings can be approximated via the eigenvalues and eigenvectors of
$\boldsymbol\psi$.

\citet{fan:2013} therefore proposed the Principal Orthogonal compleEment
Thresholding (POET) estimator of $\boldsymbol\psi$, which was motivated by the
spectral decomposition of the sample covariance matrix
\begin{equation*}
  \begin{split}
    \mathbf{S}_n & = \sum_{j=1}^J \lambda_j V_{\cdot, j} V_{\cdot, j}^\top \\
    & \approx \sum_{j=1}^L \beta_{\cdot, j}\beta_{\cdot, j}^\top +
      \sum_{j=L+1}^J \lambda_j V_{\cdot, j} V_{\cdot, j}^\top,
  \end{split}
\end{equation*}
where $\lambda_j$ and $V_{\cdot, j}$ are the $j^\text{th}$ eigenvalues and
eigenvectors of $\mathbf{S}_n$, respectively, and $\beta_{\cdot, j}$ is the
$j^\text{th}$ column of $\boldsymbol{\beta}$. For ease of notation, we denote
the second term by $\mathbf{S}_\epsilon$ and refer to this matrix as the
principal orthogonal complement. The estimator for $L$ latent variables is then
defined as

\begin{equation*}
  \hat{\Psi}_{\text{POET}}(P_n; L, s) \coloneqq \sum_{j=1}^L \lambda_j V_{\cdot,
j} V_{\cdot, j}^\top + \mathbf{T}_{\epsilon, s},
\end{equation*}
where $\mathbf{T}_{\epsilon, s}$ is the generalized thresholding matrix of
$\mathbf{S}_\epsilon$ 
\begin{equation*}
  T_{\epsilon, s}^{(jl)} \coloneqq \begin{cases}
    S_\epsilon^{(jj)}, & \text{when $j=l$} \\
    s\left(S_\epsilon^{(jl)}\right), & \text{otherwise}
  \end{cases},
\end{equation*}
for some generalized thresholding function $s$.

Although this estimator is computationally efficient, the assumptions encoding
the factor based model under which it is derived are such that the latent
features' eigenvalues grow in $J$. This results in a poor convergence rate under
the spectral norm \citep{fan:2016}.

\bibliographystyle{jasa3}
{\footnotesize \bibliography{Bibliography-MM-MC}}